\documentclass{article}
\usepackage{supplemental/bostanci-nirkhe}

\title{Separating Quantum and Classical Advice with Good Codes}
\ifsubmission
\author{Anonymous Submission to FOCS 2026}
\date{}
\else
\author[1]{John Bostanci}
\affil[1]{\small Columbia University, \textit{New York, NY}}
\author[2]{Andrew Huang}
\affil[2]{Massachusetts Institute of Technology, \textit{Cambridge, MA}}
\author[2]{Vinod Vaikuntanathan}
\date{}
\fi

\begin{document}
\maketitle

\begin{abstract}
    We show an unconditional classical oracle separation between the class of languages that can be verified using a quantum proof ($\QMA$) and the class of languages that can be verified with a classical proof ($\QCMA$).  Compared to the recent work of Bostanci, Haferkamp, Nirkhe, and Zhandry (STOC 2026), our proof is conceptually and technically simpler, and readily extends to other oracle separations. In particular, our techniques yield the {\em first unconditional} classical oracle separation between the class of languages that can be decided with quantum advice ($\BQP/\qpoly$) and the class of languages that can be decided with classical advice ($\BQP/\poly$), improving on the quantum oracle separation of Aaronson and Kuperberg (CCC 2007) and the classically-accessible classical oracle separation of Li, Liu, Pelecanos and Yamakawa (ITCS 2024).

    \vspace{4mm}
  
    Our oracles are based on the code intersection problem introduced by Yamakawa and Zhandry (FOCS 2022), combined with codes that have extremely good list-recovery properties.
\end{abstract}

\newpage
\pagenumbering{roman}
\tableofcontents
\newpage
\pagenumbering{arabic}

\section{Introduction}

We study the fundamental question of whether quantum proofs and advice are more powerful than their classical counterparts.  In the language of complexity theory, we consider whether the complexity classes $\QMA$~\cite{Kit97} and $\QCMA$~\cite{AN02}, two quantum generalizations of $\NP$ with quantum and classical proofs respectively, are distinct; and whether the complexity classes $\BQP/\qpoly$ and $\BQP/\poly$, two generalization of $\BQP$ with quantum and classical advice respectively, are distinct. 

Recall that the class $\QMA$, first defined by Kitaev~\cite{Kit97}, is the class of all languages $\Ll_{\mathsf{YES}}, \Ll_{\mathsf{NO}} \subseteq \{0,1\}^*$ for which there is a quantum polynomial-time verifier $V$ and a polynomial function $t$ such that 
\begin{align*}
  x \in \Ll_{\mathsf{YES}} & \implies \exists \ket{\psi} \in (\CC^2)^{\otimes t(\lambda)}:\ \ \Pr[V(x, \ket{\psi}) = 1] \geq {2}/{3} \hspace*{.1in} \mbox{and} \\
  x \in \Ll_{\mathsf{NO}} &  \implies \forall \ket{\psi^*} \in (\CC^2)^{\otimes t(\lambda)},\ \ \Pr[V(x, \ket*{{\psi^*}}) = 1] \leq {1}/{3}.
\end{align*}
Replacing the quantum witness $\ket{\psi}$ by a classical witness $w$ (but retaining the quantum verifier) results in the class $\QCMA$, first defined by Aharonov and Naveh~\cite{AN02}. They posed the question of whether $\QMA$ and $\QCMA$ are distinct, and it has since then been a long-standing open problem in quantum complexity theory.

Analogously, the class $\BQP/\qpoly$~\cite{NY04} is the class of all languages $\Ll_{\mathsf{YES}}, \Ll_{\mathsf{NO}} \subseteq \{0,1\}^*$ for which there is a quantum polynomial-time algorithm $A$, a polynomial function $t$ and a family of advice states $\{\ket{\adv_{\lambda}} \in (\CC^2)^{\otimes t(\lambda)}\}_{\lambda \in \mathbb{N}}$ such that for every input $x$ of length $\lambda$,
\begin{align*}
  x \in \Ll_{\mathsf{YES}} & \implies \Pr[A(x, \ket{\adv_{\lambda}}) = 1] \geq {2}/{3}, \hspace*{.1in} \mbox{and} \\
 x \in \Ll_{\mathsf{NO}}&  \implies \Pr[A(x, \ket{\adv_{\lambda}}) = 1] \leq {1}/{3}.
\end{align*}
Replacing the quantum advice $\ket{\adv_{\lambda}}$ by a classical advice string $a_{\lambda}$ (but retaining the quantumness of the algorithm $A$) results in the class $\BQP/\poly$. 
Nishimura and Yamakami~\cite{NY04} first posed the question of whether $\BQP/\qpoly$ and $\BQP/\poly$ are distinct complexity classes, and it too has been a long-standing open problem in quantum complexity theory.

Given that unconditional separations of these classes seem far out of reach of current technology as they would imply breakthrough results like $\P \neq \PSPACE$ and $\PP \subsetneq \P/\poly$, research attention has shifted to providing oracular evidence. The first progress towards this came from the seminal work of Aaronson and Kuperberg \cite{AK07}, who introduced the weaker notion of a \emph{unitary}, or \emph{quantum}, oracle as a means to separate $\QMA$ and $\QCMA$, as well as $\BQP/\qpoly$ and $\BQP/\poly$, and asked whether their quantum oracle separations could be strengthened to a classical oracle separation. After a long series of works which demonstrated either conditional or non-standard separations \cite{FK15,LLPY23,NN24,BDK24,Zha24,LMY25}, the first unconditional classical oracle separation between $\QMA$ and $\QCMA$ appeared in the recent work of Bostanci, Haferkamp, Nirkhe and Zhandry~\cite{BHNZ25}. Even so, the case of $\BQP/\qpoly$ versus $\BQP/\poly$ has remained wide open. 

The oracle used in the separation of \cite{BHNZ25} (introduced in \cite{Zha24}) leverages one of the most natural properties to separate quantum from classical computers: whether two functions are related by the Hadamard transform. This idea was first exploited in the work of \cite{Aar10}, who presented the ``Forrelation'' problem, i.e., deciding whether two functions are related by the Hadamard transform, as a candidate separation between $\BQP$ and $\PH$. The works of \cite{Zha24} and \cite{BHNZ25} lift this problem to $\QMA$ with a variant dubbed ``spectral Forrelation''. The separation of \cite{BHNZ25} employs sophisticated techniques to analyze Forrelated oracles, including ideas inspired by mathematical physics and Hamiltonian learning, resulting in a rather complicated proof. It is therefore natural to wonder whether a classical oracle separation requires such heavy technical machinery. 
\begin{quote}
\begin{center}
\textbf{Question 1:} \emph{Is there a simple classical oracle separation between $\QMA$ and $\QCMA$?}
\end{center}
\end{quote}

A closely related problem is to separate $\BQP/\qpoly$, problems that can be solved with quantum advice, from $\BQP/\poly$, problems that can be solved with classical advice.  As opposed to the proof setting, where a verifier gets an untrusted proof that can depend on their specific instance, in the advice setting the verifier receives a \emph{trusted} witness, but is expected to use their witness to solve all instances of the same length. Thus, the questions of separating $\QMA$ from $\QCMA$ and $\BQP/\qpoly$ from $\BQP/\poly$, while related, are not formally equivalent. The original work of Aaronson and Kuperberg~\cite{AK07} also gave a unitary oracle separation for $\BQP/\qpoly$ and $\BQP/\poly$ and asked if this too could be lifted to the classical oracle model. Thus, we investigate the following question:

\begin{quote}
\begin{center}
\textbf{Question 2:} \emph{Is there a (standard) classical oracle separation between $\BQP/\qpoly$ and $\BQP/\poly$?}
\end{center}
\end{quote}

The only other advice separation follows from the work of Li, Liu, Pelecanos and Yamakawa~\cite{LLPY23}, in a non-standard ``classically-accessible classical oracle'' model where all queries to the oracles are restricted to be classical. A true classical oracle separation between $\BQP/\qpoly$ and $\BQP/\poly$ remains unresolved.  


\subsection{Our Work}

We re-examine a line of works \cite{Liu23, LLPY23, BDK24} which have made progress towards a full separation between $\QMA$ and $\QCMA$, based on the code intersection problem of Yamakawa and Zhandry~\cite{YZ24}. We begin by showing that Question 1 can indeed be resolved using a modification of the code intersection problem when instantiated with codes that have extremely good list-recovery properties.

\begin{theorem}[Informal]
    There exists a classical oracle $\Oo$ such that $\QMA^{\Oo} \not\subseteq \QCMA^{\Oo}$.
\end{theorem}

Our separation, in addition to admitting a much simpler proof, has an additional benefit over the result of \cite{BHNZ25}: our more structured oracle allows us to lift our result to provide the \textbf{first unconditional classical oracle separation} between $\BQP/\qpoly$ from $\BQP/\poly$ in a straightforward manner.\footnote{A similar idea appeared in \cite{LLPY23}, where a variant of the Yamakawa-Zhandry problem was used to give a separation between quantum and classical advice for algorithms which are only allowed classical access to all oracles. As we will see later, our separation strictly improves on this result since we rule out all $\BQP/\poly$ algorithms with \emph{quantum} oracle access while only requiring a single classical query given quantum advice.} 
\begin{theorem}[Informal]
    There exists a classical oracle $\Oo$ such that $\BQP^{\Oo}/\qpoly \not\subseteq \BQP^{\Oo}/\poly$.
\end{theorem}
In fact, our oracle separates $\NP^{\Oo} \cap \coNP^{\Oo} \cap \BQP^{\Oo}/\qpoly$ from $\BQP^{\Oo}/\poly$, indicating a structural difference between our oracle and those based on spectral Forrelation \cite{Zha24,BHNZ25} or expander mixing \cite{Lut11,NN24,LMY25}, which appear to originate from $\QMA$-complete problems.

\section{Technical Overview}
\paragraph{The Yamakawa-Zhandry Algorithm.} Our separation begins with the code intersection problem~\cite{YZ24}: given a code $C \subseteq \Sigma^n = (\FF_q^s)^n$, function $H: [n] \times \Sigma \to \{0, 1\}$, and hash $x \in \{0, 1\}^n$, find a codeword $c \in C$ such that $H(i, c_i) = x_i$ for all $i \in [n]$ (we will use the shorthand $H(c) = x$ to refer to this constraint). Yamakawa and Zhandry show that when given oracle access to $H$, this problem has an efficient quantum algorithm but no (uniform) classical ones, giving a relativized separation between $\FP$ and $\FBQP$. As our result will require modifying the Yamakawa-Zhandry algorithm, we begin by briefly explaining how it works.

We begin by noting that it suffices to be able to produce the state 
\begin{equation*}\ket{\psi} \propto \sum_{\substack{\bfv: \bfv \in C \\H(\bfv) = x}} \ket{\bfv} = \sum_{\bfv \in \Sigma^n} \Id_{H, x}(\bfv) \cdot \Id_C(\bfv) \ket{\bfv}\,,
\end{equation*}
where $\Id_{H, x}(\cdot)$ and $\Id_C(\cdot)$ are indicator functions for the event $H(\bfv) = x$ and for the event $\bfv \in C$, respectively. Taking inspiration from Regev's reduction from SIS to LWE \cite{Reg09}, Yamakawa and Zhandry observe that $\ket{\psi}$ is the pointwise product of the states
\begin{equation*}
\ket{\phi_1} := \ket{\phi_1(x)} \propto \sum_{\substack{\bfv \in \Sigma^n\\ H(\bfv) = x}} \ket{\bfv} = \bigotimes_{i = 1}^n \sum_{\substack{v_i \in \Sigma\\ H(i, v_i) = x_i}} \ket{v_i} \quad \text{and} \quad \ket{\phi_2} \propto \sum_{\bfv \in C} \ket{\bfv}\,,
\end{equation*}
both of which can be prepared efficiently given access to $H$. By the convolution theorem, we know that
    \[ \QFT_q \ket{\psi} \propto \QFT_q (\ket{\phi_1} \odot \ket{\phi_2}) = \QFT_q \ket{\phi_1} \star \QFT_q \ket{\phi_2}, \]
where $\odot$ denotes the point-wise product of two vectors, and $\star$ their convolution. Therefore, it suffices to efficiently prepare the state $\ket{\mathsf{goal}} := \QFT_q \ket{\phi_1} \star \QFT_q \ket{\phi_2}$, as $\QFT_q^{-1} \ket{\mathsf{goal}} = \ket{\psi}$. If $C$ is a $\FF_q$-linear code, then $\QFT_q \ket{\phi_2}$ is simply the uniform superposition over the dual code $C^{\perp}$, so we can produce the states
\begin{align*}
   \QFT_q \ket{\phi_1} \otimes \QFT_q \ket{\phi_2} \otimes \ket{0} \propto & \sum_{\substack{\ee \in \Sigma^n\\ \bfv \in C^{\perp}}} \sqrt{\Dd_{H, x}(\ee)} \ket{\ee} \ket{\bfv} \ket{0}
    \\\stackrel{U_{\mathsf{add}}}{\mapsto} &\sum_{\substack{\ee \in \Sigma^n\\ \bfv \in C^{\perp}}} \sqrt{\Dd_{H, x}(\ee)} \ket{\ee} \ket{\bfv} \ket{\bfv+\ee}
    \\
    \stackrel{U_{\mathsf{sub}}}{\mapsto} & \sum_{\substack{\ee \in \Sigma^n\\ \bfv \in C^{\perp}}} \sqrt{\Dd_{H, x}(\ee)} \ket{0} \ket{\bfv} \ket{\bfv+\ee},
\end{align*}
while our desired state is
    \[ \ket{0} \ket{0} \ket{\mathsf{goal}} = \ket{0} \ket{0} \otimes \QFT_q \ket{\phi_1} \star \QFT_q \ket{\phi_2} \propto \sum_{\substack{\ee \in \Sigma^n\\ \bfv \in C^{\perp}}} \sqrt{\Dd_{H, x}(\ee)} \ket{0} \ket{0} \ket{\bfv+\ee}, \]
where $\Dd_{H, x}(\cdot)$ is the density function of $\QFT_q \ket{\phi_1}$ (here, we are ignoring phases for simplicity of exposition). For a completely random function $H$ and a fixed $x$, we observe that $\Dd_{H, x}(\cdot)$ will have roughly half of its weight on 0 in each coordinate (since around half of all symbols in each coordinate should hash to 0 or 1), while the remaining half of its weight will be close to uniform over nonzero symbols. By taking $C$ to be a folded Reed-Solomon (FRS) code with sufficiently high rate, Yamakawa and Zhandry show that one can efficiently decode $C^{\perp}$ from errors over $\Dd_{H, x}(\cdot)$ with high probability and consequently prepare $\ket{\mathsf{goal}}$ as desired.

\paragraph{Lifting the Separation.} The obvious issue in using the code intersection problem to separate $\QMA$ from $\QCMA$, however, is that both $\QMA$ and $\QCMA$ algorithms can make quantum oracle access to $H$, and our problem is already in $\BQP$! Thus, we must make some modifications to the problem at hand.

First, as observed by \cite{Liu23}, the first phase of this algorithm can be made non-adaptive: the state $\ket{\phi_2}$ depends only on $C$, while $\ket{\phi_1}$ depends only mildly on $x$, since we can simply prepare all $2n$ preimage states
\begin{equation*} \ket{H^{-1}_1(0)} := \sum_{\substack{x \in \Sigma\\ H(1, x) = 0}} \ket{x},\ \ldots \ , \ket{H^{-1}_n(1)} := \sum_{\substack{x \in \Sigma\\ H(n, x) = 1}} \ket{x}\,,
\end{equation*}
before selecting $\ket{\phi_1}$ when given $x$. On the other hand, the second phase does not require access to $H$, and should work equally well for all $x$. Therefore, if we are given the state
    \[ \ket{\adv_H} := \bigotimes_{i = 1}^n \ket{H^{-1}_i(0)} \otimes \bigotimes_{i = 1}^n \ket{H^{-1}_i(1)} \otimes \ket{\phi_2} \]
as our quantum proof or advice, we can produce $\ket{\phi_1(x)} := \bigotimes_{i = 1}^n \ket{H^{-1}_i(x_i)}$ for {\em any} $x$ and use it to produce a solution to the code intersection problem for $x$ without access to $H$! We can therefore replace $H$ with the much weaker oracle $O_H(x, v)$ that simply verifies if the vector $v$ hashes to $x$. 

Of course, this problem is still in $\NP \subseteq \QCMA$, since for a fixed $x$, the prover can always send any codeword $v$ that hashes to $x$! Luckily, we claim that this is pretty much the \emph{only} thing that the prover can do. To operationalize this intuition, we note that our quantum proof/advice is in some sense encoding many codewords along with their hashes in superposition. Thus, taking some inspiration from \cite{LLPY23, BDK24}, we define the \emph{Code Intersection Subset Size} ($\mathsf{CISS}$) problem as follows: 
\begin{quote} 
    Estimate the size of a set $E \subseteq \{0, 1\}^n \times \Sigma^n$, under the promise that either $E = \{0, 1\}^n \times \Sigma^n$ or $|E| \leq t$ for some threshold $t \ll 2^n$, given access to the following oracle $O[H, E](x, v)$:
    \begin{eqnarray*}
        O[H, E](x, v) & = \left\{ \begin{array}{ll} 1 & \mbox{if $v \in C$, $H(v) = x$, and $(x, v) \in E$,} \\ 0 & \mbox{otherwise.} \end{array} \right.
    \end{eqnarray*}
\end{quote} 

This problem is naturally in $\QMA$: the prover can give as proof $\ket{\adv_H}$ which depends only on $H$. Given $\ket{\adv_H}$, the verifier can sample a random $x \in \{0, 1\}^n$ and should be able to produce with high probability $v \in C$ such that $H(v) = x$. The oracle $O$ therefore allows the verifier to check if $(x, v) \in E$, making the set estimation problem trivial. In fact, the $\QMA$ verifier need not be concerned about malicious proofs, since a NO instance of $O$ always outputs 0 whenever $x$ (which is sampled solely by the verifier) is not in the support of $E$. 

We note that our problem differs from the oracle problems defined in \cite{LLPY23, BDK24}, which take $E = F \times \Sigma^n$ for some small set $F$. This modification, while not impacting our $\QMA$ algorithm, will be crucial in establishing a $\QCMA$ lower bound.  

\paragraph{An Entropic Viewpoint.} To rule out $\QCMA$ proof systems, we first observe that the major difference between quantum and classical proofs lies in their \emph{clonability}. In particular, an (oracle) algorithm which uses a \emph{classical} witness can always be re-run with the same witness, even if it makes measurements. This simple and seemingly obvious fact, first formally identified in \cite{Zha24}, was utilized to great effect by \cite{BHNZ25} to give their classical oracle separation, and we will take advantage of it as well.

To this end, suppose there was some $Q$-query $\QCMA$ verifier $V$ which succeeded in the $\mathsf{CISS}$ problem. We observe that this means that $V$ can always distinguish between $O[H, \{0, 1\}^n \times \Sigma^n]$ and $O[H, E]$ whenever $E$ is small. But these oracles differ only at inputs $(x, v)$ where $v \in C$, $H(v) = x$, and $(x, v) \notin E$! Thus, by the hybrid lemma \cite{BBBV97}, if we measure a random query that $V$ makes to $O[H, E]$, we should expect to get a new pair $(x, v) \notin E$ such that $v \in C$ and $H(v) = x$ with good probability provided $E$ is small.\footnote{Critically, if we have already successfully guessed some collection $E$ of code words and hash values, we can always perfectly simulate $O[H, E]$. The same is \emph{not} true for the oracles of \cite{LLPY23, BDK24}, where even ``small" sets correspond to exponentially many codewords.} This gives rise to a natural algorithm for \emph{guessing} the hash values of codewords: starting with $E = \emptyset$, simulate a run of $V$ with $O[H, E]$ and measure a random query before adding the measurement outcome to $E$; rinse and repeat. By our previous argument, conditioned on having a good witness, each iteration of this algorithm should correctly produce a new codeword and hash with non-negligible probability. We can therefore turn $V$ into a guesser with non-uniform advice which correctly produces many distinct codewords and their hash values without making \emph{any} oracle queries. After guessing the classical $w$-bit witness, this gives rise to an unconditional no-query algorithm which guesses the hash values of $\ell$ codewords in $C$ for all $\ell \leq t$ with probability 
    \[ 2^{-w} \cdot \bigg(\Omega(Q^{-2})\bigg)^{\ell} = 2^{-\poly(n)} \cdot \bigg(\frac{1}{\poly(n)} \bigg)^{\ell}~. \]
We now argue that this is in fact impossible. Observing that $H$ is independently random at each coordinate $i$ and symbol  $x \in \Sigma$, we can upper bound the success probability of any sampler which produces $\ell$ points by $(\frac{1}{2})^{s(\ell)}$, where $s(\ell)$ is the minimum number of symbols that appear among $\ell$ distinct codewords in $C$. Taking $\ell = t = \omega(\poly(n))$, if we can argue that $s(\ell) = \omega(\log n \cdot \ell)$, then we see that
    \[ 2^{-\poly(n)} \cdot \left(\frac{1}{\poly(n)}\right)^{\ell} = \bigg(\frac{1}{\poly(n)}\bigg)^{\ell} \gg \left(\frac{1}{2^{\omega(\log n)}}\right)^{\ell} = \left(\frac{1}{2}\right)^{s(\ell)}, \]
which will give us our desired contradiction.

\paragraph{List Recovery and Code Expansion.} How might we bound $s(\ell)$? For any $\ell$ distinct codewords $c_1, \ldots, c_{\ell} \in C$, define the lists $S_1, \ldots, S_n \subseteq \Sigma$ such that $S_i$ consists of all symbols in $\Sigma$ that appear in the $i$'th coordinate of some codeword $c_j$ for $j \in [\ell]$. Clearly, $s(\ell) = \min_{c_1, \ldots, c_{\ell}} \sum_{i = 1}^n |S_i|$, so there must be lists $S^{*}_1, \ldots, S^{*}_n$ such that
    \[ |C \cap (S^{*}_1 \times \ldots \times S^{*}_n)| \geq \ell \quad \text{and} \quad \sum_{i = 1}^n |S^{*}_i| = s(\ell)~. \]
We now see that the question of how small $s(\ell)$ can be is precisely characterized by the \emph{list-recoverability} of $C$. In particular, if we know that $C$ is $(L, O(L))$-list-recoverable for $L \lesssim \ell$, then this would mean that $\frac{s(\ell)}{n} = \frac{1}{n} \sum_{i = 1}^n |S^{*}_i| = \Omega(\ell)$ as desired!

Note that this is a pretty strong condition; it necessitates the use of codes that have {\em near-optimal list recovery}, a property that in particular is not satisfied by the FRS codes used by \cite{YZ24} (see \Cref{subsec:simplifications} for further discussion). Fortunately, there is a fix: the setting of zero-error list recovery is closely linked to the notion of unbalanced expanders, and the recent work of \cite{KTS22} shows that multiplicity codes exhibit precisely the sort of list recovery that we need. Moreover, the fact that multiplicity codes are $\FF_q$-linear and that their duals have relatively good distance \cite{RZVW24} should guarantee the success of our $\QMA$ algorithm. We note that although our dual code happens to admit efficient unique decoding, our separation only needs $C^{\perp}$ to be \emph{statistically} uniquely decodable, since we can always provide an (inefficient) decoding oracle.
 
\paragraph{A Final Complication.} It seems that this rather simple argument completes our separation; after all, by switching to using multiplicity codes (rather than folded Reed-Solomon codes as in \cite{YZ24,LLPY23,BDK24}), we have been able to rule out all possible $\QCMA$ algorithms. Sadly, we have to deal with one final and rather subtle issue, which has to do with the parameters of the multiplicity codes: in the process of obtaining excellent list recovery/expansion from our multiplicity codes, we are forced to make the relative rate of our (primary) code sub-constant, which means our dual code now has sub-constant relative distance! Recalling that our error distribution should concentrate on vectors with Hamming weight roughly $n/2$, we observe that this level of noise is now likely intolerable as there may not even exist a unique decoding most of the time under this error distribution.

Our solution is relatively simple, and it uses the generous amount of flexibility that the \cite{Reg09,YZ24} algorithm affords us. Instead of using a completely random function $H$, we will instead make our function \emph{biased} in favor of 0 (reminiscent of recent strategies employed by \cite{GGJL25, GGKM} in the context of communication complexity and low-depth implementations of the Yamakawa-Zhandry algorithm). That is, for each element $\sigma \in \Sigma$, $H(i, \sigma)$ will take on the value 0 with probability $p \gg \frac{1}{2}$. Thus, for $x = 0^n$, we can expect the error vectors in $\QFT \ket{\phi_1}$ to have Hamming weight $\approx n(1-p)$, drastically reducing the amount of noise that we are required to decode from with respect to $C^{\perp}$. 

This change does not come for free, however: unlike in \cite{GGJL25}, where the goal was to invert only $x = 0^n$, we need to be able to invert many $x$'s, including those with large Hamming weight. By biasing $H$ towards 0, on inputs like $x = 1^n$, we create an error distribution which has expected Hamming weight $\approx np \gg n/2$, thereby \emph{worsening} our ability to invert! 

Our final idea is to observe that since our algorithm can only invert low Hamming-weight vectors $x$ rather than all vectors in $\{0, 1\}^n$, we can simply \emph{modify the problem to enforce this condition}. Instead of trying to invert all $x \in \{0, 1\}^n$, we can focus on inverting vectors of the form $x \| 0^{n-n^c}$, where $x \in \{0, 1\}^{n^c}$ and $0 < c \ll 1$. That is, we will now try to differentiate $E = \{0, 1\}^{n^c} \times 0^{n-n^c} \times \Sigma^n$ from $E \subseteq F \times 0^{n-n^c} \times \Sigma^n$ where $|F| \leq t \ll 2^{n^c}$. Since $x \| 0^{n-n^c}$ has Hamming weight at most $n^c$, the corresponding error distribution for $\ket{\phi_1(x \| 0^{n-n^c})}$ will be concentrated on vectors with Hamming weight at most $n^c + (n-n^c)(1-p) \approx n(1-p)$, guaranteeing the success of our $\QMA$ algorithm provided our dual code has distance at least $O(n(1-p))$ and $p \lesssim 1-1/n^{1-c}$.

\paragraph{A General Recipe for a Classical Oracle Separation.}

Before describing the advice separation, we summarize all of the steps we have taken so far to provide a general recipe for getting a $\QMA$ versus $\QCMA$ oracle separation. We start with an infinite family of linear codes $\{C_{\lambda}\}_{\lambda \in \NN}$ over a large alphabet $\Sigma$ (so that there are many solutions to the code intersection problem) such that:
\begin{enumerate}
    \item Codewords of $C_{\lambda}$ consist of $n = \poly(\lambda)$ many symbols from $\Sigma$,
    \item $C_{\lambda}^{\perp}$ can be (efficiently or combinatorially) decoded from up to $\Omega(\lambda^{1+c})$ errors for some constant $c$,
    \item $C_{\lambda}$ has near optimal list recovery for sufficiently large $\lambda$; i.e. $s(\ell) = \Omega(n \cdot \ell)$ for $\ell = \lambda^{\omega(1)}$.
\end{enumerate}
Then for every $\lambda$, we can sample hash functions $H_1, \ldots, H_{n}$ to be biased so that roughly a $\lambda / n(\lambda)$-fraction of symbols are pre-images of $1$. We will ask for pre-images of $x \| 0^{n - \lambda}$ for $x \in \{0, 1\}^{\lambda}$. The bias of the $H_i$, together with the fact that we ask for a hash that has at most $\lambda$ many $1$'s, ensures that the dual decoding problem encounters an error with $O(\lambda)$ Hamming weight with high probability, which falls under our dual decoding distance of $\Omega(\lambda^{1+c})$. Thus, the Yamakawa-Zhandry algorithm works and the problem stays in $\QMA$. At the same time, a $\QCMA$ verifier will imply a sampler that outputs $v$ codewords of the code with probability $\poly(\lambda)^{-v}$, and list recovery will enforce that this corresponds to $\Omega(v \cdot n)$ symbols.  The bias of the $H_i$ will mean that the probability of guessing all symbols correctly will be $\left(1 - \lambda/n\right)^{\Omega(v \cdot n)} \approx \exp(-\lambda v) \ll \poly(\lambda)^{-v}$, giving us a contradiction.   

By instantiating this recipe with multiplicity codes, we arrive at our classical oracle separation.  

\paragraph{Moving to the Advice Setting: $\BQP/\qpoly$ vs. $\BQP/\poly$.} As mentioned previously, a separation in the proof setting does not immediately imply an advice-style separation. In fact, the $\mathsf{CISS}$ problem is easy with {\em trusted} advice: a single classical bit suffices to indicate whether the set in question is large or small.

The general paradigm we follow begins with an idea of \cite{AK07}: after sampling a random binary language $\Ll \subseteq \{0, 1\}^n$, the \cite{AK07} oracle essentially outputs whether an instance $x$ is in the language if it is also given as input a specific quantum state $\ket{\psi}$. Intuitively, if it is hard to find $\ket{\psi}$ given classical advice, then it should be difficult to decide $\Ll$. In this manner, \cite{AK07} transform a hard quantum search problem into a hard decision problem (relative to a quantum oracle). The classical oracle separation of \cite{BHNZ25}, being a kind of ``classicalization'' of the Aaronson-Kuperberg oracle, has a similar quantum search problem associated with it. However, since a classical oracle can no longer directly check the answer to a quantum search problem, obtaining a classical oracle separation between $\BQP/\qpoly$ and $\BQP/\poly$ seems to necessitate the use of a hard \emph{classical search problem} instead.

Here, the relative simplicity of our separation and our use of a classical search/$\TFNP$ problem allows us to straightforwardly extend our results to the question of $\BQP/\qpoly$ and $\BQP/\poly$. In contrast, it is not at all obvious (to us) how to construct even a candidate separating language based on the (decisional) spectral Forrelation problem of \cite{BHNZ25}.

To be more specific, the code intersection problem gives rise to the following $\BQP/\qpoly$ language: begin by sampling some random binary language $\Ll \subseteq \{0, 1\}^n$. On input $x \in \{0, 1\}^n$, our oracle $O$ will return whether $x \in \Ll$ or $x \notin \Ll$ provided it is also given a codeword $v$ which hashes to $x$. On the one hand, our original $\QMA$ algorithm still works as a $\BQP/\qpoly$ machine, since $\ket{\adv_H}$ is agnostic of $x$ and allows us to produce some $v$ for any $x$. On the other hand, as we have argued earlier, even trusted classical advice should not help a $\BQP$ machine to produce many valid codewords, so for most $x$, any $\BQP$ machine with classical advice will not be able to receive the output of the oracle indicating whether $x \in \Ll$. Since the advice is bounded, it cannot itself describe many elements of $\Ll$, so the $\BQP$ machine will fail to decide whether $x \in \Ll$ for most $x$. The problem of guessing the value of a random function $H$ given bounded-size advice and without querying $H$ can be made precise by appealing to results on Yao's box problem due to Chung, Guo, Liu and Qian~\cite{CGLQ20}. We show that by modifying ideas used by \cite{LLPY23} for analyzing {\em classical-access} $\BQP/\poly$ algorithms, we can reduce a decisional advice separation to a search-based separation for bona fide $\BQP/\qpoly$ machines. Finally, by adapting the techniques we used to separate $\QMA$ from $\QCMA$, we can also rule out the existence of $\BQP/\poly$ machines which solve the code intersection search problem, completing the advice separation. 


\section{Discussion and Open Questions}\label{sec:discussion}

Before describing the technical details of our construction and proof, we take some time to discuss parts of our result that might be interesting to a reader.  One question a reader might have is why this result, and similar results, did not appear sooner.  One of the main ideas of recent works on separating $\QMA$ from $\QCMA$ has been to start from a refined reduction from a $\QCMA$ algorithm to a sampler.  The idea of using sampling to separate $\QMA$ from $\QCMA$, although not new to this work, nor \cite{BHNZ25}, only started appearing in recent works~\cite{NZ23,Zha24}. Before this recent line of work, the common strategy for starting a separation was to take the ``most popular witness'' approach outlined by \cite{AK07}, which is used in \cite{FK15,NN24,BDK24,LMY25}.  Thus, one might expect that it would have been difficult to find this oracle separation prior to that.  However, the related idea of reducing bounds on advice to bounds on a ``multi-instance'' variant of some problem has previously appeared in literature on the auxiliary-input random oracle model (AI-ROM), in particular being used to rule out advice proofs for function inversion~\cite{CGLQ20}, and appeared in works as early as Aaronson's result on the power of one-way quantum communication~\cite{Aar04}.

These auxiliary-input random oracle model proofs bear some resemblance to the sampler idea. For example, if a $T$-query algorithm $\Aa$ inverts a function $f$ using an $S$-bit \emph{classical} witness $w$ with probability $\delta$, then it is not hard to see that there is a $gT$-query algorithm $\Bb$ which outputs preimages for $g$ random images of $f$ with no advice and probability $2^{-S} \cdot \delta^g$. The approach is (more or less) the one taken by \cite{Aar04,CGLQ20} for classical advice. Clearly, such an approach suffices for classical advice states, since such algorithms can always be rewound. Here, the no-cloning theorem appears to present a barrier for quantum advice, but there is a workaround: using multiple copies of advice, the adversary's success probability can be boosted to $1-\negl(\lambda)$, allowing its output to be copied and the algorithm to be rewound without causing damage to the advice state by the gentle measurement lemma. Thus, the quantum witness can be re-used \emph{without} cloning it, and so the same upper bound should apply.

So why does such an argument not actually apply to our situation (and therefore fail to produce a separation)? A more careful examination reveals that the approach with quantum advice relies on the assumption that the aforementioned measurement produces a unique output with high probability! This is fine if one only needs to determine \emph{if} an adversary has produced a preimage of a function, for example, but not \emph{what} that preimage is. For example, an adversary could use its advice state to produce a superposition over many preimages before eventually measuring to produce one of exponentially many possible preimages (which is in fact what our $\QMA$ verifier does!). Then, although the measurement which verifies if the adversary can invert is gentle, the measurement which outputs a preimage as well is not. Thus, relative to the works on the AI-ROM, the main conceptual shift in recent sampler results is to rely on natural measurements which are inherently ``destructive''. This turns out to be enough to force an algorithm solving a multi-instance problem to clone their advice state, thereby distinguishing between quantum and classical advice. Readers familiar with previous literature may also note that unlike works in the AI-ROM, our reduction (as well as the reduction of \cite{BHNZ25}) produces a multi-instance adversary which makes \emph{no} oracle queries.\footnote{Technically speaking, the reduction of \cite{BHNZ25} begins with an adversary $\Aa^{U}$ whose description depends on an oracle $U$ and makes queries to an oracle $S$, and produces an adversary $\Bb^{U}$ which no longer queries $S$ but still depends in some manner on the oracle $U$.} 

Beyond the sampler idea, one might ask whether other approaches to separating $\QMA$ from $\QCMA$ using a code intersection oracle might have been successful.  We emphasize that while the code intersection problem has previously been considered in the context of proof and advice separations \cite{Liu23, LLPY23, BDK24}, all existing works employed FRS codes, which (as mentioned earlier) do not enjoy strong enough list-recovery properties. In particular, running our argument with FRS codes gives an upper bound on the sampling success probability which is far too weak to derive \emph{any} meaningful conclusions! On the other hand, multiplicity codes do not appear to have strong enough decoding properties to handle the large amounts of noise that would be incurred by perfectly random functions $H$, which necessitates our use of biased oracles and restriction to low Hamming-weight vectors (an idea which did not appear in earlier, more limited, separations).

\paragraph{Other remarks.}

As in \cite{BHNZ25}, it is not hard to see that our oracles also separate the clonable variants of $\QMA$ and $\BQP/\qpoly$ from their regular counterparts. 

Finally, although we primarily employed biased random oracles to make decoding possible for low-distance codes, we note that (as first observed by \cite{GGKM}) biasing $H$ has the additional benefit of making the Yamakawa-Zhandry algorithm \emph{much more efficient}. Concretely, if we tweak $H$ to have bias $p = 2/3$ or $p = 3/4$, this already decreases the noise level sufficiently that we can rely on (significantly faster and simpler) \emph{unique decoders} rather than list decoders, which currently appears to be the major bottleneck in runtime.

We conclude this section with a discussion of some broader open questions left open by this work. 

\subsection{Structure versus randomness in classical oracle separations}

As stated in the introduction, one major barrier in lifting the oracle separation of \cite{BHNZ25} to the advice setting is that the oracle is most naturally associated with a hard quantum search problem, instead of a classical one. In their classical oracle separation, the authors identify a way of sampling instances of the spectral Forrelation problem such that the pair of functions that seems completely random, except that they are related by the Fourier transform (and, in the case of \cite{BHNZ25}, one of the two oracles being sparse).  In some informal sense, our oracle separation enables an efficient quantum verifier to extract more information (namely, the solution to a hard search problem) from its witness, but doing so seems to require additional structure in the oracle.    

A natural question to ask is how much information can be encoded into a quantum state before it becomes clonable, and whether our ideas are useful in encoding information into other kinds of quantum states. The cryptographic analogy of a separation between $\QMA$ and $\QCMA$ (or, really $\mathsf{UnclonableQMA}$) is a primitive called quantum money~\cite{Aar09}.  These are states that can not be cloned, but can be verified, similar to witnesses for $\QMA$ problems that are not in $\QCMA$. An extremal form of this primitive (and this idea of encoding information in a quantum state) is known as ``copy-protected software'', wherein an efficient quantum party can extract the input-output behavior of an entire classical function from a quantum state. Currently, there are several candidates for quantum money constructions in the plain model~\cite{FGHL12,BNZ25,Zha25}, but less is known about copy-protected software. We hope that ideas from our separation might be useful in finding such constructions. To make progress towards cryptographic instantiations, one concrete direction to explore is a precise characterization of which witnesses cause our $\QMA$ verifier to accept with high probability.

\subsection{\texorpdfstring{$\QMA$-completeness of a decoding problem}{QMA-completeness of a decoding problem}}
One interesting difference between our oracle separation and the oracle separation of \cite{BHNZ25} is that their oracle separation can be seen as an obfuscation of a $\QMA$-complete problem. To elaborate further, just as problems involving random sparse functions might model the difficulty of constructing a SAT solver which ignores the structure of the SAT instance it receives, the spectral Forrelation problem models an algorithm for solving a two-basis local Hamiltonian problem that does not look at the structure of the two local Hamiltonians it receives. Note that despite the connection to a $\QMA$-complete problem, this property is not actually needed to achieve a separation between $\QMA$ and $\QCMA$, as highlighted in the actual oracle separation of \cite{BHNZ25}. In particular, the ``YES'' and ``NO'' instances can be taken to have sets $S$ of a fixed size $\ell$ or $\leq \ell/10$. Such \emph{distributions} of oracles can be easily distinguished by an $\AM$ protocol \cite{GS86}, but it appears unlikely that solving the actual spectral Forrelation problem can be done in $\AM$, because another way to sample ``NO'' instances of the spectral Forrelation problem would be to re-sample sets $S'$ independent of $U$. Distinguishing such pairs from Forrelated pairs $(S, U)$ seems to truly require a $\QMA$ verifier.

In contrast, the problem we construct really is in $\AM$, as it directly involves distinguishing between large and small sets, and any reasonable variant of the code intersection problem would likely remain in the polynomial hierarchy. Of course, an oracle separation between $\QCMA$ and $\AM$ already exists (and in fact, comes from early work on separating $\QMA$ from $\QCMA$~\cite{FK15}!), but we find it intriguing that our separation relies on a problem which appears to be of only intermediate difficulty and does not need the ``full'' power of $\QMA$ in some sense. In fact, as our problem appears on its face to be completely unrelated to quantum algorithms, we believe it remains an extremely interesting question to find a $\QMA$-complete variant of the code intersection problem.

\subsection{Simplifications to the separation}\label{subsec:simplifications}
Naturally, one may ask if our separations can be made \emph{even simpler}; here, we outline a few directions to consider. 

\paragraph{On Round-Reduction Arguments.} 
The work of \cite{BDK24} can be thought of as a round-reduction argument as follows: in the style of \cite{AK07}, we will, given any $\QCMA$ proof system, fix some classical witness which corresponds to the largest set of NO instances. The polynomial bound on the size of the proof means that the collection of functions $H$ which are consistent with this witness remains substantial, and in particular the distribution over all consistent $H$ must have large min-entropy. We conclude that the set of symbols $S$ which have low entropy conditioned on this classical proof is also bounded.

We now look at the verifier's inputs to the oracle in the first round of queries. Observe that any inputs corresponding to non-codewords, as well as ones which have low overlap with $S$ will be correct with negligible probability. By the hybrid lemma, it suffices to restrict our attention to codewords which have high overlap with $S$ (so called ``dangerous'' inputs) -- but this set is bounded precisely by the list recovery of our code! At this stage, we can simply give away the values of $H$ on dangerous inputs \emph{for free}, allowing the verifier to simulate its first round of queries. As a consequence, we can peel off a round of queries at the cost of requiring more advice. Repeating this ``peeling and bloating" routine with FRS codes results in a $o(\log n/\log \log n)$-round bound.

Sadly, this approach seems fairly doomed if we stick to the Yamakawa-Zhandry problem: even with nearly optimal list recovery, our advice will certainly increase by some constant factor with each round, so we would remain stuck at an $o(n)$-round bound. One could imagine a better argument that uses some special property of any collection of recovered codewords that ensures these lists do not grow by too much iteratively, but this seems quite difficult (and would definitely be much more complicated!).

\paragraph{Other Codes?}
One might wonder why we need multiplicity codes here -- after all, there are other codes which give rise to unbalanced expanders besides those of \cite{KTS22}, namely the constructions of \cite{GUV09}, but examining these constructions in closer detail presents some unexpected issues:
\begin{enumerate}
    \item \cite{GUV09} construct two unbalanced expanders with near-optimal expansion, one based on Parvaresh-Vardy codes and the other based on a subcode of FRS codes. Alas, neither of these instantiations are linear, which means our quantum proof/advice-based algorithm will fail.
    \item \cite{GUV09} also considers unbalanced expanders based on plain FRS codes, which \emph{are} linear in some sense, but such expanders have expansion which is too weak to show meaningful classical hardness. 
\end{enumerate}

\paragraph{On the Necessity of Near-Optimal List Recovery.} Our proof relies on fairly strong expansion/list recovery properties of the underlying code $C$. It is not hard to see that we can tolerate slightly suboptimal expansion, i.e. if $s(\ell) = \Omega(\ell/\poly(n))$. Extending our sampling-based argument to work for polynomial or super-polynomial expansion, i.e. $s(\ell) = \ell^{O(1)}$, would open the door to a much larger class of usable codes. 

\paragraph{What About Random Linear Codes?} 
It is well known that random linear codes (RLCs) and their duals have good distance with high probability, and they are by definition linear, which means our quantum proof/advice-based algorithm will succeed. However, to show classical hardness, we need near-optimal list recovery, which remains a challenging open problem to show for RLCs \cite{LS25}. We believe that a proof of such a result for RLCs is the clearest way to conceptually simplify our separation.

\section{Preliminaries}

\subsection{Notation}
We say that a function $\delta: \mathbb{N} \mapsto [0, 1]$ is inverse polynomial if there exists a polynomial $p$ such that $\delta(n) \leq 1/p(n)$ for sufficiently large $n$. A function $\epsilon: \mathbb{N} \mapsto [0, 1]$ is negligible if for every polynomial $p$, for all sufficiently large $n$, $\epsilon(n) < 1/p(n)$. 

We use the notation $\mathrm{id}$ to denote the identity operator.  We will occasionally concatenate superscripts when it is clear from context, so $\mathrm{QFT}^{-1, \otimes n}$ denotes $(\mathrm{QFT}^{-1})^{\otimes n}$. We will also sometimes abbreviate the tensor product state $\ket{0}^{\otimes n}$ as $\ket{0^n}$. 

A register $\mathsf{R}$ is a named finite-dimensional complex Hilbert space. If $\mathsf{A}, \mathsf{B}, \mathsf{C}$ are registers, for example,
then the concatenation $\mathsf{A}\mathsf{B}\mathsf{C}$ denotes the tensor product of the associated Hilbert spaces. For a linear transformation $L$ and register $\mathsf{R}$, we write $L_{\mathsf{R}}$ to indicate that $L$ acts on $\mathsf{R}$, and similarly we write $\sigma_{R}$ to indicate that a state $\sigma$ is in the register $\reg{R}$. 


\subsection{Probability and Complexity Theory}
\begin{definition}[Modified from \cite{BHNZ25}]
    We use the phrase ``an oracle'' to refer to a function $\Oo∶ \{0, 1\}^{*} \to \{0, 1\}$. A quantum query algorithm is a quantum circuit that interacts with an oracle $\Oo∶ \{0, 1\}^{*} \to \{0, 1\}$ via a query gate $\ket{x, b} \to \ket{x, b \oplus \Oo(x)}$, which acts on an $|x|+1$-qubit query register for various $x$. The algorithm is described by an alternating sequence of unitaries (drawn from any fixed gate set) and query gates. After all gates are applied, a designated qubit is measured in the standard basis to determine acceptance. The circuit has ancilla qubits which are initialized to $\ket{0}$ and all intermediate unitaries may act on an arbitrary (but finite) number of qubits.

    A quantum query algorithm may also receive an auxiliary witness or advice as input. A quantum witness/advice state is a state $\ket{\psi}$ on some (finite) number of qubits, while a classical witness/advice string is a (finite) bitstring $w$, treated as a computational basis state. The algorithm's acceptance probability may depend on both the oracle and the witness.
\end{definition}

\begin{definition}[\cite{BHNZ25}]
    We denote a family of quantum oracle circuits/algorithms by $\{\Aa_{\lambda}\}_{\lambda \in \NN}$, where the index $\lambda$ corresponds to the length of the explicit input to the computational problem. $\{\Aa_{\lambda}\}_{\lambda \in \NN}$ is $\P$-uniform if there exists a deterministic polynomial-time Turing machine $M$ that, on input $1^{\lambda}$, outputs a full classical description of the circuit $\Aa_{\lambda}$. The runtime of $M$ implies that $\Aa_{\lambda}$ has at most $\poly(\lambda)$ gates (oracle or elementary), queries $\Oo$ at lengths of at most $\poly(\lambda)$, and receives witnesses/advice of length at most $\poly(\lambda)$.
\end{definition}
We can now define the complexity classes that we will consider in this work.
\begin{definition}[Oracle $\QCMA$]
    A promise language $\Ll^{\Oo} = (\Ll_{\mathsf{yes}}, \Ll_{\mathsf{no}})^{\Oo} \subseteq \{0, 1\}^{*}$ is in $\QCMA^{\Oo}$ if there exists a $\P$-uniform family of quantum oracle circuits $\Aa_{\lambda}$ with $\Aa_{\lambda}$ accepting a witness of length $t(\lambda)$, such that for every input $x$ of length $\lambda = |x|$,
    \begin{itemize}
        \item $x \in \Ll_{\mathsf{yes}} \implies \exists w \in \{0, 1\}^{t(\lambda)} \text{ s.t. } \Pr[\Aa_
        {\lambda}^{\Oo}(x, w) = 1] \geq \frac{2}{3}$,
        \item $x \in \Ll_{\mathsf{no}} \implies \forall \tilde{w} \in \{0, 1\}^{t(\lambda)}, \Pr[\Aa_
        {\lambda}^{\Oo}(x, \tilde{w}) = 1] \leq \frac{1}{3}$.
    \end{itemize}
\end{definition}
\begin{definition}[Oracle $\QMA$]
    A promise language $\Ll^{\Oo} = (\Ll_{\mathsf{yes}}, \Ll_{\mathsf{no}})^{\Oo} \subseteq \{0, 1\}^{*}$ is in $\QMA^{\Oo}$ if there exists a $\P$-uniform family of quantum oracle circuits $\Aa_{\lambda}$ with $\Aa_{\lambda}$ accepting a witness of length $t(\lambda)$, such that for every input $x$ of length $\lambda = |x|$,
    \begin{itemize}
        \item $x \in \Ll_{\mathsf{yes}} \implies \exists \ket{\psi} \in (\CC^2)^{\otimes t(\lambda)} \text{ s.t. } \Pr[\Aa_
        {\lambda}^{\Oo}(x, \ket{\psi}) = 1] \geq \frac{2}{3}$,
        \item $x \in \Ll_{\mathsf{no}} \implies \forall \ket*{\tilde{\psi}} \in (\CC^2)^{\otimes t(\lambda)}, \Pr[\Aa_
        {\lambda}^{\Oo}(x, \ket*{\tilde{\psi}}) = 1] \leq \frac{1}{3}$.
    \end{itemize}
\end{definition}
\begin{definition}[Oracle $\BQP/\poly$]
    A promise language $\Ll^{\Oo} = (\Ll_{\mathsf{yes}}, \Ll_{\mathsf{no}})^{\Oo} \subseteq \{0, 1\}^{*}$ is in $\BQP^{\Oo}/\poly$ if there exists a $\P$-uniform family of quantum oracle circuits $\Aa_{\lambda}$ with $\Aa_{\lambda}$ accepting advice of length $t(\lambda)$ and an advice family $\{\adv_{\lambda}\}_{\lambda \geq 1}$ where $|\adv_{\lambda}| = t(\lambda)$, such that for every input $x$ of length $\lambda = |x|$,
    \begin{itemize}
        \item $x \in \Ll_{\mathsf{yes}} \implies \Pr[\Aa_
        {\lambda}^{\Oo}(x, \adv_{\lambda}) = 1] \geq \frac{2}{3}$,
        \item $x \in \Ll_{\mathsf{no}} \implies \Pr[\Aa_
        {\lambda}^{\Oo}(x, \adv_{\lambda}) = 1] \leq \frac{1}{3}$.
    \end{itemize}
    If $t(\lambda) = 0$ then we say $\Ll^{\Oo} \in \BQP^{\Oo}$.
\end{definition}
\begin{definition}[Oracle $\BQP/\qpoly$]
    A promise language $\Ll^{\Oo} = (\Ll_{\mathsf{yes}}, \Ll_{\mathsf{no}})^{\Oo} \subseteq \{0, 1\}^{*}$ is in $\BQP^{\Oo}/\qpoly$ if there exists a $\P$-uniform family of quantum oracle circuits $\Aa_{\lambda}$ with $\Aa_{\lambda}$ accepting advice of length $t(\lambda)$ and an advice family $\{\ket{\adv_{\lambda}}\}_{\lambda \geq 1}$ where $\ket{\adv_{\lambda}} \in (\CC^2)^{\otimes t(\lambda)}$, such that for every input $x$ of length $\lambda = |x|$,
    \begin{itemize}
        \item $x \in \Ll_{\mathsf{yes}} \implies \Pr[\Aa_
        {\lambda}^{\Oo}(x, \ket{\adv_{\lambda}}) = 1] \geq \frac{2}{3}$,
        \item $x \in \Ll_{\mathsf{no}} \implies \Pr[\Aa_
        {\lambda}^{\Oo}(x, \ket{\adv_{\lambda}}) = 1] \leq \frac{1}{3}$.
    \end{itemize}
\end{definition}
Given a quantum query algorithm, we can define the query mass of the algorithm on a particular set of inputs.  

\begin{definition}[Query mass]\label{def:query_mass}
For an oracle circuit $A$ making $Q$ quantum queries to an oracle $\Oo$ with input domain $D$, let $\sum_x \alpha^{(i)}_x \ket{x}\ket{\psi_{x}}$ be the state of the algorithm immediately before their $i$'th query to $\Oo$, where the first register is the input register to the oracle, and let $M_x(i) = |\alpha^{(i)}_x|^2$ be the \emph{query mass} of $x$ in the $i$'th query. For a subset $V \subseteq [Q] \times D$, let $M_V = \sum_{(i, x) \in V} M_x(i)$ be the \emph{total query mass} of points in $V$.
\end{definition}
The following theorem was proven in \cite{BBBV97}, using the hybrid method.

\begin{theorem}[Hybrid method~\cite{BBBV97}]\label{thm:bbbv}
    Let $A$ be an oracle circuit which makes $Q$ queries to an oracle $\Oo$ with input domain $D$. If we modify $\Oo$ into an oracle $\Oo'$ which differs only on a set of time-input pairs $V \subseteq [Q] \times D$, then $$|\Pr[A^{\Oo}(\cdot) = 1] - \Pr[A^{\Oo'}(\cdot) = 1]| \leq 4\sqrt{QM_V}~.$$
\end{theorem}

Finally, we will also use some basic probability lemmas.
\begin{lemma}[Chernoff Bound]\label{lemma:chernoff}
Let $X_1, \ldots, X_n$ be independent random variables taking values in $\{0, 1\}$, $X := \sum_{i = 1}^n X_i$, and $\mu := \Exp[X]$. For any $\delta \geq 0$, it holds that $\Pr[X \geq (1 + \delta) \mu] \leq e^{-\delta^2\mu/(2+\delta)}$.
\end{lemma} 

\begin{lemma}[Borel–Cantelli, \cite{Bor09, Can17}]\label{lemma:borel_cantelli}
    Let $\{X_{\lambda}\}_{\lambda \in \NN}$ be a sequence of (not necessarily independent) random variables with values in $\{0, 1\}$. If $\sum_{\lambda = 1}^{\infty} \Exp[X_{\lambda}] < \infty$, then $ \Pr\left[\sum_{\lambda = 1}^{\infty} X_\lambda = \infty\right] = 0$.
\end{lemma}

\subsection{Coding Theory}
For a prime power $q$, we denote by $\FF_q$ the finite field of order $q$ and denote by $(\FF_q)_{< k}[X]$ the set of univariate polynomials over $\FF_q$ with degree less than $k$.
\begin{definition}
    A code of length $n \in \NN$ over an alphabet $\Sigma$ is a subset $C \subseteq \Sigma^n$. $C \subseteq \Sigma^n$ is said to be \emph{$\FF_q$-linear} if its alphabet $\Sigma = \FF_q^s$ for some field $\FF_q$ and a positive integer $s \geq 1$ and $C$ is an $\FF_q$-linear subspace of $\Sigma^n$. Equivalently, this means that for any two codewords $x, y \in C$ and scalar $\alpha \in \FF_q$, both $x+y$ and $\alpha \cdot x$ are in $C$.
    
    For an $\FF_q$-linear code $C$, the dual code of $C$ is the code $C^{\perp} \subseteq (\FF_q^s)^n$ containing all strings $c' \in (\FF_q^s)^n$ which satisfy $$\sum_{i = 1}^n \sum_{j = 1}^s (c'_i)_j \cdot (c_i)_j = 0$$ for all $c \in C$. Observe that $C^{\perp}$ is always $\FF_q$-linear, and that $|C| \cdot |C^{\perp}| = |\Sigma|^n$ if and only if $C$ is $\FF_q$-linear.
\end{definition}
For any vector $x \in \Sigma^n$, define $\hw(x) \in [0, n]$ as the Hamming weight of $x$, i.e. the number of nonzero elements in $x$. We say that $C \subseteq \Sigma^n$ has distance $d$ if for any two distinct codewords $c_1, c_2 \in C$, $\hw(c_1-c_2) \geq d$.

\begin{definition}[Formal and Hasse derivatives]
Let $f(X) = \sum_{j = 0}^n a_j X^j \in \FF_q[X]$ be a univariate polynomial over $\FF_q$. We define the $i$'th formal and Hasse derivatives of $f(X)$ as the linear operators which take $f(X)$ to the polynomials $$f^{[i]}(X) = \sum_{j = i}^n \frac{j!}{(j-i)!} a_j X^{j-i} \hspace{10pt} \mbox{  and  } \hspace{10pt} f^{(i)}(X) = \sum_{j = i}^n \binom{j}{i} a_j X^{j-i}~,$$ 
respectively. Note that $f^{[i]}(X) = i! f^{(i)}(X)$ for all $i$ and $f$.
\end{definition}

\begin{definition}[Univariate multiplicity codes, from \cite{RT97, Nie01, KSY14}]
    Let $\FF_q$ be a finite field and let $s$ be a positive integer. Let $\alpha_1, \ldots, \alpha_n$ be distinct points in $\FF_q$, and let $k < sn$ be a positive integer. The univariate multiplicity code $\Mult_{s, \FF_q}(\alpha_1, \ldots, \alpha_n; k)$ is the code over the alphabet $\Sigma = \FF_q^s$ of length $n$ which associates each polynomial $f(X) \in (\FF_q)_{<k}[X]$ to the codeword $c \in \Sigma^n$ such that for $i \in [n]$, 
        \[ c_i = (f^{(0)}(\alpha_i), f^{(1)}(\alpha_i), \ldots, f^{(s−1)}(\alpha_i)), \]
    where $f^{(j)}$ is the $j$'th \textbf{Hasse derivative} of $f$. Let $\Mult_{s, \FF_q, k} := \Mult_{s, \FF_q}(1, \ldots, q; k)$; note that $\Mult_{s, \FF_q, k}$ is $\FF_q$-linear.
\end{definition}

\begin{definition}[List recoverability]
    A code $C \subseteq \Sigma^n$ is \emph{$(\ell, L)$-list recoverable} if for all $S_1, \ldots, S_n \subseteq \Sigma$ such that $\frac{1}{n} \sum_{i = 1}^n |S_i| \leq \ell$,
        \[ |\{(x_1, \ldots, x_n) \in C: \forall i \in [n], x_i \in S_i\}| \leq L. \]
\end{definition}

\begin{definition}[Expanders \cite{GUV09}]
    A bipartite graph with $N$ left-vertices, $M$ right-vertices, and left-degree $D$ is specified by a function $\Gamma : [N] \times [D] \to [M]$, where $\Gamma(x, y)$ denotes the $y$’th neighbor of $x$. For a set $X \subseteq [N]$, we write $\Gamma(X)$ to denote its set of neighbors $\bigcup_{x \in X, y \in [D]} \Gamma(x, y)$. For a set $T \subseteq [M]$, we write $\LIST_{\Gamma}(T) = \{ x: \Gamma(x) \subseteq T \}$. 
    
    We say that $\Gamma$ is a \emph{$(K, A)$-expander} if for every set $X \subseteq [N]$ of size at most $K$, $|\Gamma(X)| \geq A \cdot |X|$. Note that if $\Gamma$ is a $(K, A)$-expander then for all $B \leq K$ and all sets $T$ such that $|T| < AB$, $|\LIST_{\Gamma}(T)| < B$.
\end{definition}
Our separation will utilize expanders based on multiplicity codes as constructed in \cite{KTS22}.\footnote{We note that the proof of expansion extends straightforwardly to subgraphs of $\Gamma$ defined by taking edges corresponding to subsets $S \subseteq \FF_q$, although it suffices for us to take $S = \FF_q$.}
\begin{theorem}[\cite{KTS22}]\label{thm:unbalanced_expanders}
    For every field $\FF_q$, $k$, $s \in \NN$ such that $15 \leq s+1 \leq k \leq \mathsf{char}(\FF_q)$, identify the elements of $\FF_q^k$ with univariate polynomials of degree less than $k$. Define the graph $\Gamma: \FF_q^k \times \FF_q \to \FF_q^{s+1}$ by 
        \[ \Gamma(f, y) = (y, f^{[0]}(y), f^{[1]}(y), \ldots, f^{[s-1]}(y)), \]
    where $f^{[i]}$ is the $i$’th \textbf{formal derivative} of $f$ in $\FF_q[X]$. For every $K > 0$, $\Gamma$ is a $(K, A)$-expander where
        \[ A = q - \frac{k(s+1)}{2} \cdot (qK)^{\frac{1}{s+1}}. \]
\end{theorem}

\begin{corollary}
    For each security parameter $\lambda \in \NN$, let $k = \lambda^3$, $\lambda^5 < q \leq 2\lambda^5$ be any prime, and $s = \lambda$. Identify the elements of $\FF_q^k$ with univariate polynomials of degree less than $k$. Define the code $C'_{\lambda} \subseteq \Sigma^q = (\FF_q^s)^q$ with encoding map
    \begin{equation*} 
        \Enc(f) = \{(f^{[0]}(y), f^{[1]}(y), \ldots, f^{[s-1]}(y))\}_{y \in \FF_q}\,,
    \end{equation*}
    where $f^{[i]}$ is the $i$’th \textbf{formal derivative} of $f$. Then, for sufficiently large $\lambda$, $C'_{\lambda}$ is $(\ell, 2\ell)$-list recoverable if $\ell \leq 2^s = 2^{\lambda}$.
\end{corollary}
\begin{proof}
    Set $K = 2^{s+1}$ and fix $\ell \leq 2^s$; \Cref{thm:unbalanced_expanders} implies that $\Gamma: \FF_q^k \times \FF_q \to \FF_q^{s+1}$ is a $(K, A)$-expander where
        \[ A \geq q-\frac{k(s+1)}{2} \cdot (qK)^{\frac{1}{s+1}} = q-\frac{k(s+1)}{2} \cdot (q \cdot 2^{s+1})^{\frac{1}{s+1}} > q/2,  \]
    for sufficiently large $\lambda$.

    Fix any lists $S_1, \ldots, S_q \subseteq \Sigma$ such that $\frac{1}{q} \sum_{i = 1}^q |S_i| \leq \ell$ and consider the set
        \[ R = \bigcup_{i = 1}^q R_i, \hspace{10pt} R_i := \{ (i, w_0, \ldots, w_{s-1}): (w_0, \ldots, w_{s-1}) \in S_i \}. \]
    Observe that $R$ is a set of right vertices of $\Gamma$ and that $|R| = \sum_i |S_i| \leq q \ell < A \cdot 2\ell \leq AK$. Now consider the set $X$ of all polynomials $f$ such that for all $i \in [q], \Enc(f)_i \in S_i$; our goal is to bound $|X|$. By construction, $X = \LIST_{\Gamma}(R)$, so it follows that $|X| = |\LIST_{\Gamma}(R)| < 2\ell$.
\end{proof}

In our parameter regime, it is straightforward to see that the code $C_{\lambda} = \Mult_{s, \FF_q, k}$ has the same list-recoverability as $C'_{\lambda}$, since for fields $\FF_q$ where $\mathsf{char}(\FF_q) > s$, $C_{\lambda}$ and $C'_{\lambda}$ are identical up to scalar factors.
\begin{corollary}\label{cor:mult_list_recovery}
    For each security parameter $\lambda \in \NN$, let $k = \lambda^3$, $\lambda^5 < q \leq 2\lambda^5$ be any prime, and $s = \lambda$. Then, for sufficiently large $\lambda$, $C_{\lambda} = \Mult_{s, \FF_q, k}$ is $(\ell, 2\ell)$-list recoverable for all $\ell \leq 2^{\lambda}$.
\end{corollary}

Finally, we will use a result about duals of univariate multiplicity codes which we reprove in \Cref{app:mult}.\footnote{We note that although the work of \cite{RZVW24} was recently retracted, the particular theorem we use remains correct. For completeness, an entirely self-contained proof of this fact is given in the appendix.}
\begin{theorem}[\cite{RZVW24}]\label{thm:mult_dual}
    For all parameters $s$, $q$, and $k < sq$, $(\Mult_{s, \FF_q, k})^{\perp}$ has distance at least $\frac{k+1}{s}$.
\end{theorem}

\subsection{Yao's Box Problem and Non-Uniform Advice}
We will need the following results on non-uniform advice for our separation between $\BQP/\qpoly$ and $\BQP/\poly$.
\begin{theorem}[\cite{CGLQ20}]\label{thm:yao_box}
    Let $G: [N] \to \{0, 1\}$ be a random function. Let $\Aa$ be an unbounded-time algorithm, with $S$ bits of classical advice $z_G$. For an index $x \in [N]$, let $G|^{x}: N \to \{0, 1\}$ denote the function that results from removing $x$ from $G$; in other words, on inputs $x' \neq x$, $G(x') = G|^{x}(x')$ and $G|^{x}(x) = 0$. The probability that $\Aa$ computes $G(x)$ while making $Q$ quantum queries to $G|^{x}$ for a random index $x$ is at most
        \[ \Pr_{G, x}[\Aa^{G|^{x}}(z_G, x) = G(x)] \leq \frac{1}{2}+O\left(\frac{(S+\log N)Q}{N}\right)^{1/3}. \]
\end{theorem}
\begin{lemma}\label{lemma:advice_query_measurement}
    Let $G : \{0, 1\}^{\lambda} \to \{0, 1\}$ be a uniformly random function. For an algorithm $\Aa$ that makes $Q(\lambda) = \poly(\lambda)$ quantum queries to $G$ and a family of $t(\lambda) = \poly(\lambda)$-bit classical advice $\{z_G\}_G$, suppose that
        \[ \Pr_{G, x \gets \{0, 1\}^{\lambda}}[\Aa^G(z_G, x) = G(x)] > \frac{3}{5}. \]
    Then, for sufficiently large $\lambda$, for a $\frac{1}{4000Q^2}$ fraction of $x \in \{0, 1\}^{\lambda}$, measuring a random query of $\Aa^{G}$ (for randomly sampled $G$) will produce $x$ with probability at least $\frac{1}{3200Q^2}$.
\end{lemma}
\begin{proof}
    The proof closely follows \cite{LLPY23}, but considers quantum queries instead of classical ones. The only way for $\Aa$ to distinguish $G$ from $G|^{x}$ is to have nontrivial query mass at $x$. Denote by $M_{G, x}$ the total query mass that $x$ is placed by $\Aa$ when querying $G$. For each $G$ and $x$ we have that
        \[ |\Pr[\Aa^G(z_G, x) = G(x)]-\Pr[\Aa^{G|^x}(z_G, x) = G(x)]| \leq 4\sqrt{QM_{G, x}}. \]
    Now we consider the case when we uniform randomly choose $x \gets \{0, 1\}^{\lambda}$, and require $\Aa^{G|^x}(z_G, x)$ to output $G(x)$. This is exactly Yao’s box problem, so by \Cref{thm:yao_box},
    \begin{align*}
        \Pr_{G, x}[\Aa^{G|^{x}}(z_G, x) = G(x)] \leq \frac{1}{2}+O\left(\frac{(t+\lambda)Q}{2^\lambda}\right)^{1/3} &= \frac{1}{2}+\negl(\lambda) \\
        \implies \Pr_{G, x}[\Aa^G(z_G, x) = G(x)]-\Pr_{G, x}[\Aa^{G|^x}(z_G, x) = G(x)] &\geq \frac{1}{10}-\negl(\lambda).
    \end{align*}
    Thus, 
    \begin{align*}
        \Exp_{G, x}[M_{G, x}] &\geq \Exp_{G, x}\left[\frac{1}{16Q}\left(\Pr[\Aa^G(z_G, x) = G(x)]-\Pr[\Aa^{G|^x}(z_G, x) = G(x)]\right)^2\right] \\
        &\geq \frac{1}{16Q} \left(\Exp_{G, x}[\Pr[\Aa^G(z_G, x) = G(x)]-\Pr[\Aa^{G|^x}(z_G, x) = G(x)]]\right)^2 \geq \frac{1}{1600Q}-\negl(\lambda),
    \end{align*}
    by Jensen's inequality. Finally, by a Markov inequality, we see that
        \[ \Pr_x\left[\Exp_{G}[M_{G, x}] \geq \frac{1}{3200Q}\right] \geq \frac{1}{3200Q^2}-\negl(\lambda). \]
    Thus, for sufficiently large $\lambda$, for a $\frac{1}{3200Q^2}-\negl(\lambda) \geq \frac{1}{4000Q^2}$ fraction of $x \in \{0, 1\}^{\lambda}$, measuring a random query of $\Aa^{G}$ (for randomly sampled $G$) will produce $x$ with probability at least $\frac{1}{3200Q^2}$.
\end{proof}

\section{The Generalized Code Intersection Problem}
\subsection{Definitions and Basic Facts}
We begin by recalling the definitions and basic results from \cite{YZ24}.  Much of this section will be taken directly from \cite{YZ24}, with only minor modifications. We first define the code intersection relation, which is essentially the problem of finding codewords over $n$ symbols whose symbols have a particular hash value.  
\begin{definition}[Code intersection relation, adapted from \cite{YZ24, LLPY23, BDK24}]
    For a function $H : [n] \times \Sigma \to \{0, 1\}$ and a code $C \subseteq \Sigma^n$, define the code intersection relation $R_{C, H} \subseteq \{0, 1\}^n \times \Sigma^n$ by
        \[ R_{C,H} = \{ (\x, \bfv) = (x_1, \ldots, x_n, v_1, \ldots, v_n) : ((v_1, \ldots, v_n) \in C) \land (\forall i \in [n], H(i, v_i) = x_i) \}. \]
\end{definition}
\begin{remark}
    We can view $H : [n] \times \Sigma \rightarrow \{0, 1\}$ as a collection of $n$ many functions, $H(i, \cdot): \Sigma \to \{0, 1\}$, and we will at times use the notation $H_i: \Sigma \to \{0, 1\}$ when referring to the function corresponding to the $i$'th output coordinate of $H$.
\end{remark}

\begin{definition}[Trace over a finite field~\cite{YZ24}]
    For any prime power $q = r^m$ where $r$ is prime, we define the trace function $\Tr(x) := \sum_{i = 0}^{m-1} x^{r^i}$ which maps elements of $\FF_q$ to $\FF_r$. The trace function is $\FF_r$-linear: for all $a, b \in \FF_r$ and $x, y \in \FF_q$, $\Tr(ax+ by) = a\Tr(x) + b\Tr(y)$. In addition, for any $x \in \FF_q^n$, $\sum_{y \in \FF_q^n} \omega_r^{\Tr(x \cdot y)} = 0$, where $\omega_r := e^{2\pi i/r}$.
\end{definition}
\begin{definition}[Quantum Fourier transform over a finite field~\cite{YZ24}]
    For a finite field $\FF_q$ where $q = r^m$ and $r$ is prime, the quantum Fourier transform over $\FF_q$ is the unitary denoted by $\QFT_{q}$ such that for any $x \in \FF_q$,
        \[ \QFT_{q} \ket{x} = \frac{1}{\sqrt{q}} \sum_{z \in \FF_q} \omega_r^{\Tr(x \cdot z)} \ket{z}. \]
    The QFT over an alphabet $\Sigma = \FF_q^s$ is the $s$-wise tensor product of $\QFT_{q}$: for $\x = (x_1, \ldots, x_s) \in \Sigma$,
        \[ \QFT_{\Sigma} \ket{\x} := \QFT^{\otimes s}_{q} \ket{x_1} \ldots \ket{x_s} = \frac{1}{\sqrt{|\Sigma|}} \sum_{\z \in \Sigma} \omega_r^{\Tr(\x \cdot \z)} \ket{\z}. \]
    Similarly, for any positive integer $n$ and $\x \in \Sigma^n$, we have
        \[ \QFT_{\Sigma}^{\otimes n} \ket{\x} = \frac{1}{|\Sigma|^{n/2}} \sum_{\z \in \Sigma^n} \omega_r^{\Tr(\x \cdot \z)} \ket{\z}. \]
    The unitary $\QFT_q$ can be approximated within error $\eps$ in operator norm in time $\poly(\log q, \log 1/\eps)$ \cite{CW02, vDHI06}.
\end{definition}
\begin{definition}[Fourier transform of a function~\cite{YZ24}]
    For functions $f, g: \Sigma^n \to \CC$, we define
        \[ \widehat{f}(\z) := \frac{1}{|\Sigma|^{n/2}} \sum_{\x \in \Sigma^n} f(\x) \omega_r^{\Tr(\x \cdot \z)}, \hspace{10pt} (f \cdot g)(\x) := f(\x) \cdot g(\x), \hspace{5pt} \text{and} \hspace{5pt} (f \star g)(\x) := \sum_{\y \in \Sigma^n} f(\y) \cdot g(\x-\y). \]
    Note that
        \[ \QFT_{\Sigma}^{\otimes n} \sum_{\x \in \Sigma^n} f(\x) \ket{\x} = \sum_{\z \in \Sigma^n} \widehat{f}(\z) \ket{\z}. \]
\end{definition}
\begin{fact}[\cite{YZ24}]\label{lemma:fourier_properties}
    The following properties hold for the Fourier transform:
    \begin{enumerate}
        \item (Parseval's equality) For all functions $f: \Sigma^n \to \CC$, $\sum_{\x \in \Sigma^n} |f(\x)|^2 = \sum_{\z \in \Sigma^n} |\widehat{f}(\z)|^2$.
        \item (Pointwise transform) Suppose that we have $f_i: \Sigma \to \CC$ for $i \in [n]$ and $f: \Sigma^n \to \CC$ is defined by $f(\x) := \prod_{i = 1}^n f_i(x_i)$. Then $\widehat{f}(\z) = \prod_{i = 1}^n \widehat{f}_i(z_i)$.
        \item (Convolution theorem) For all functions $f, g, h: \Sigma^n \to \CC$, $\widehat{f \cdot g} = \frac{1}{|\Sigma|^{n/2}}(\widehat{f} \star \widehat{g})$, $\widehat{f \star g} = |\Sigma|^{n/2}(\widehat{f} \cdot \widehat{g})$, and $\widehat{f \cdot (g \star h)} = (\widehat{f} \star (\widehat{g} \cdot \widehat{h}))$.
    \end{enumerate}
\end{fact}

\begin{lemma}[Fourier transform of a linear code]\label{lemma:dual_fourier}
    Let $C \subseteq \Sigma^n = (\FF_q^s)^n$ be any $\FF_q$-linear code. Then,
        \[ f(\bfu) = \begin{cases} \frac{1}{\sqrt{|C|}} & \text{if $\bfu \in C$,} \\ 0 & \text{otherwise.} \end{cases} \hspace{5pt} \Longleftrightarrow \hspace{5pt} \widehat{f}(\bfu) = \begin{cases} \frac{1}{\sqrt{|C^{\perp}|}} & \text{if $\bfu \in C^{\perp}$,} \\ 0 & \text{otherwise}. \end{cases} \]
\end{lemma}
\begin{proof}
Since $C$ is $\FF_q$-linear, $|C| \cdot |C^{\perp}|= |\Sigma|^n$, and thus for any $\z \in C^{\perp}$,
\begin{align*}
    \widehat{f}(\z) = \frac{1}{|\Sigma|^{n/2}} \sum_{\bfu \in \Sigma^n} f(\bfu) \omega_r^{\Tr(\bfu \cdot \z)} = \frac{1}{|\Sigma|^{n/2}} \sum_{\bfu \in C} \frac{1}{\sqrt{|C|}} = \frac{1}{\sqrt{|C^{\perp}|}}.
\end{align*}
Finally, $\widehat{f}(\z) = 0$ for $\z \notin C^{\perp}$ by Parseval's equality. The reverse direction follows by an identical argument.
\end{proof}

\begin{lemma}\label{lemma:euclidean_unnormalized}
    Let $\ket{\psi}$, $\ket{\phi}$ be states such that $\norm{\ket{\psi}} = 1$ and $\norm{\ket{\psi}-\ket{\phi}} \leq \eps < 1$. Then, $\norm{\ket{\psi}-\frac{\ket{\phi}}{\|\ket{\phi}\|}} \leq 2\eps$. 
\end{lemma}
\begin{proof}
    By the reverse/inverse triangle inequality, 
        \[ |\|\ket{\phi}\| - 1| = |\|\ket{\phi}\| - \|\ket{\psi}\|| \leq \|\ket{\phi}-\ket{\psi}\| = \eps \implies \|\ket{\phi}\| \geq 1-\eps > 0. \]
    We can thus define the normalized state $\ket{\phi'} := \frac{\ket{\phi}}{\|\ket{\phi}\|}$. By the regular triangle inequality,
        \begin{equation*}
            \|\ket{\psi}-\ket{\phi'}\| \leq \|\ket{\psi}-\ket{\phi}\| + \|\ket{\phi}-\ket{\phi'}\| \leq \eps + |\|\ket{\phi}\|-1| \leq 2\eps. \qedhere
        \end{equation*}
\end{proof}
\begin{lemma}[\cite{BV93}]\label{lemma:euclidean_to_tv}
    Let $\ket{\psi}, \ket{\phi}$ be states such that $\|\ket{\psi}\| = \|\ket{\phi}\| = 1$ and $\|\ket{\psi}-\ket{\phi}\| \leq \eps$. Then the total variation distance between the probability distributions resulting from measurements of $\ket{\phi}$ and $\ket{\psi}$ is at most $4\eps$.
\end{lemma}

\newcommand{\GOOD}{\mathsf{GOOD}}
\newcommand{\BAD}{\mathsf{BAD}}
\newcommand{\Good}{\mathsf{Good}}
\newcommand{\Bad}{\mathsf{Bad}}

Now we state the main algorithmic result of \cite{YZ24}, namely that a quantum algorithm can approximately implement the convolution trick for some families of functions. 
\begin{lemma}[\cite{YZ24}]\label{lemma:yz_regev}
    Let $\ket{\psi}$ and $\ket{\phi}$ be quantum states on a quantum system over an alphabet $\Sigma = \FF_q^s$ written as
    \begin{align*}
        \ket{\psi} = \sum_{\bfu \in \Sigma^n} V(\bfu) \ket{\bfu} \hspace{10pt} \text{and} \hspace{10pt} \ket{\phi} = \sum_{\ee \in \Sigma^n} W(\ee) \ket{\ee},
    \end{align*}
    for functions $V, W: \Sigma^n \to \CC$. Let $F : \Sigma^n \to \Sigma^n$ be a function and let $\GOOD \subseteq \Sigma^n \times \Sigma^n$ be a subset such that for any $(\bfu, \ee) \in \GOOD$, we have $F(\bfu + \ee) = \bfu$. Define $\BAD := (\Sigma^n \times \Sigma^n) \setminus \GOOD$. Suppose that
    \begin{align*}
        \sum_{(\bfu, \ee) \in \BAD} |\widehat{V}(\bfu) \widehat{W}(\ee)|^2 \leq \eps \hspace{10pt} \text{and} \hspace{10pt} \sum_{\z \in \Sigma^n} \left|\sum_{(\bfu, \ee) \in \BAD: \bfu+\ee = \z} \widehat{V}(\bfu)\widehat{W}(\ee)\right|^2 &\leq \delta\,.
    \end{align*}
    Define the unitaries $U_{\mathsf{add}}$ and $U_F$ as follows:
    \begin{align*}
        U_{\mathsf{add}} = \sum_{(\bfu, \ee) \in \Sigma^{n} \times \Sigma^{n}} \ket{\bfu + \ee, \bfu} \!\! \bra{\ee, \bfu}\quad\text{and} \quad U_F = \sum_{(\bfu, \mathbf{w}) \in \Sigma^{n} \times \Sigma^{n}} \ket{\mathbf{w}, \bfu - F(\mathbf{w})}\!\!\bra{\mathbf{w}, \bfu}\,.
    \end{align*}
    Then,
    \begin{equation*}
        \norm{\left(\QFT^{-1, \otimes n}_{\Sigma} \otimes \mathrm{id}\right) \cdot U_F \cdot U_{\mathsf{add}} \cdot \left(\QFT^{\otimes n}_{\Sigma} \otimes \QFT^{\otimes n}_{\Sigma}\right) \ket{\phi} \ket{\psi} - |\Sigma|^{n/2} \sum_{\z \in \Sigma^n} (V(\z) \cdot W(\z)) \ket{\z} \ket{0}} \leq \sqrt{\eps} + \sqrt{\delta}\,.
    \end{equation*}
\end{lemma}
\Cref{lemma:euclidean_unnormalized,lemma:euclidean_to_tv,lemma:yz_regev} imply the following corollary (as the QFT, $U_{\mathsf{add}}$, and $U_F$ are unitaries).
\begin{corollary}\label[corollary]{cor:measuring_distance}
    Let $V$ and $W$ be functions, and $\epsilon$ and $\delta$ be the corresponding error parameters from \Cref{lemma:yz_regev}. For any property $\Pp: \Sigma^n \to \{0, 1\}$, if $\|\ket{\phi}\| = \|\ket{\psi}\|= 1$ and measuring (in the standard basis) the normalization of
    \begin{equation*}
        |\Sigma|^{n/2} \sum_{\z \in \Sigma^n} (V(\z) \cdot W(\z)) \ket{\z} \ket{0}\,,
    \end{equation*}
    produces an outcome $\z$ such that $\Pp(\z) = 1$ with probability $p$, then measuring (in the standard basis)
    \begin{equation*}
        \left(\QFT^{-1, \otimes n}_{\Sigma} \otimes \mathrm{id}\right) \cdot U_F \cdot U_{\mathsf{add}} \cdot \left(\QFT^{\otimes n}_{\Sigma} \otimes \QFT^{\otimes n}_{\Sigma}\right) \ket{\phi} \ket{\psi}\,,
    \end{equation*}
    produces an outcome $\z$ such that $\Pp(\z) = 1$ with probability at least $p-8(\sqrt{\eps}+\sqrt{\delta})$.
\end{corollary}

\subsection{Technical Lemmas}\label{subsec:technical_lemmas}
In this section, we state some technical lemmas to extend the Yamakawa-Zhandry algorithm to work with biased oracles. We first introduce a pair of functions $V$ and $W$ that represent normalized indicators for a code $C$ and the preimages of any output $\bb$ for a function $H$. For any $\FF_q$-linear code $C \subseteq \Sigma^n = (\FF_q^s)^n$, function $H: [n] \times \Sigma \to \{0, 1\}$, and string $\bb \in \{0, 1\}^n$, let $V: \Sigma^n \to \CC$, $W_i^{H_i, b_i}: \Sigma \to \CC$, and $W^{H, \bb}: \Sigma^n \to \CC$ be defined as follows:
\begin{align*}
    V(\bfu) &= \begin{cases}
        \frac{1}{\sqrt{|C|}}& \text{$\bfu \in C$}\\
        0 & \text{otherwise}
    \end{cases}\,, \\
    W_i^{H_i, b_i}(e) &= \begin{cases}
        \frac{1}{\sqrt{|T_i^{H_i, b_i}|}}& \text{$e \in T_i^{H_i, b_i}$}\\
        0 & \text{otherwise}
    \end{cases}\,, \\
    W^{H, \bb}(e_1, \ldots, e_n) &= \prod_{i = 1}^n W_i^{H_i, b_i}(e_i)\,,
\end{align*} 
where $T_i^{H_i, b_i} \subseteq \Sigma$ is the subset consisting of $e_i \in \Sigma$ such that $H_i(e_i) = b_i$.

\newcommand{\Bias}{\mathsf{Bias}}
\begin{definition}[$p$-biased distribution]\label{def:yz_distribution}
    For any $p \in [0, 1]$ and set $\Sigma$, let $\Bias_{p, \Sigma}$ denote the distribution over functions from $\Sigma$ to $\{0, 1\}$ that samples $F: \Sigma \to \{0, 1\}$ with probability $p^{|F^{-1}(1)|}(1-p)^{|F^{-1}(0)|}$. Let $\Bias_{n, p, \Sigma}$ denote the distribution over functions $G: [n] \times \Sigma \to \{0, 1\}$ that samples $G$ with probability $\Bias_{n, p, \Sigma}(G) := \prod_{i = 1}^n \Bias_{p, \Sigma}(G_i)$.
\end{definition}

The following claim follows immediately from the definition of $\Bias_{n, p, \Sigma}$.
\begin{claim}\label{claim:perms}
    Let $\pi$ be any permutation over $\Sigma$ (resp. $\Sigma^n$). Then, the distributions $\Bias_{p, \Sigma}$ and $\Bias_{p, \Sigma} \circ \pi$ (resp. $\Bias_{n, p, \Sigma}$ and $\Bias_{n, p, \Sigma} \circ \pi$) are identical.
\end{claim}

The following lemma shows that when we take the Fourier transform of the preimage state of $H_i$ sampled from $\Bias_{p, \Sigma}$, the resulting Fourier coefficients are in expectation uniform over all non-zero elements of $\Sigma$ and fixed (either $p$ or $1-p$ depending on if we are taking the preimage of $0$ or $1$) on $0$.  
\begin{claim}\label{claim:fourier_coeffs}
    Fix any string $\bb \in \{0, 1\}^n$. For all $i \in [n]$ and $\sigma, \sigma' \in \Sigma \setminus \{0\}$, it holds that
        \[ \Exp_{H_i \gets \Bias_{p, \Sigma}}[|\widehat{W}_i^{H_i, b_i}(0)|^2] = \begin{cases} p & \text{if $b_i = 1$} \\ 1-p & \text{if $b_i = 0$} \\ \end{cases} \hspace{10pt} \text{and} \hspace{10pt} \Exp_{H_i \gets \Bias_{p, \Sigma}}[|\widehat{W}_i^{H_i, b_i}(\sigma)|^2] = \Exp_{H_i \gets \Bias_{p, \Sigma}}[|\widehat{W}_i^{H_i, b_i}(\sigma')|^2]. \]
\end{claim}
\begin{proof}
    We can directly compute the expected Fourier weight on $0$ as follows:
    \begin{align*}
        \Exp_{H_i \gets \Bias_{p, \Sigma}}\left[\left|\widehat{W}_i^{H_i, b_i}(0)\right|^2\right] = \Exp_{H_i \gets \Bias_{p, \Sigma}}\left[\left|\frac{1}{\sqrt{|\Sigma|}} \sum_{z \in \Sigma } W_i^{H_i, b_i}(z)\right|^2\right] = \frac{\Exp_{H_i}\left[\left|T_i^{H_i, b_i}\right|\right]}{|\Sigma|} = \begin{cases} p & \text{if $b_i = 1$} \\ 1-p & \text{if $b_i = 0$} \\ \end{cases}\,.
    \end{align*} 
    Since $\sigma \neq 0$ (resp. $\sigma' \neq 0$), for any $w \in \FF_q$, the number of $z \in \Sigma$ such that $\sigma \cdot z = w$ is $|\Sigma|/q$. Therefore, there is a permutation $\pi_{e, e'} : \Sigma \to \Sigma$ such that $\sigma \cdot z = \sigma' \cdot \pi_{\sigma, \sigma'}(z)$ for all $z \in \Sigma$. Thus, by \Cref{claim:perms},
    \begin{align*}             
        \Exp_{H_i \gets \Bias_{p, \Sigma}}\left[\left|\widehat{W}_i^{H_i, b_i}(\sigma)\right|^2\right] &= \Exp_{H_i \gets \Bias_{p, \Sigma}}\left[\left|\frac{1}{\sqrt{|\Sigma|}} \sum_{z \in \Sigma} W_i^{H_i, b_i}(z) \cdot \omega_r^{\Tr(\sigma \cdot z)}\right|^2\right] \\
        &= \Exp_{H_i \gets \Bias_{p, \Sigma}}\left[\left|\frac{1}{\sqrt{|\Sigma|}} \sum_{z \in \Sigma} W_i^{H_i \circ \pi^{-1}_{\sigma, \sigma'}, b_i}(\pi_{\sigma, \sigma'}(z)) \cdot \omega_r^{\Tr(\sigma' \cdot \pi_{\sigma, \sigma'}(z))}\right|^2\right] \\
        &= \Exp_{H_i \gets \Bias_{p, \Sigma}}\left[\left|\frac{1}{\sqrt{|\Sigma|}} \sum_{z \in \Sigma } W_i^{H_i \circ \pi^{-1}_{\sigma, \sigma'}, b_i}(z) \cdot \omega_r^{\Tr(\sigma' \cdot z)}\right|^2\right] \\
        &= \Exp_{H_i \gets \Bias_{p, \Sigma}}\left[\left|\frac{1}{\sqrt{|\Sigma|}} \sum_{z \in \Sigma } W_i^{H_i, b_i}(z) \cdot \omega_r^{\Tr(\sigma' \cdot z)}\right|^2\right] \\
        &= \Exp_{H_i \gets \Bias_{p, \Sigma}}\left[\left|\widehat{W}_i^{H_i, b_i}(\sigma')\right|^2\right]\,. \qedhere
    \end{align*}
\end{proof}

\noindent
For any function $\Dec_{C^{\perp}}: \Sigma^n \to \Sigma^n$, define the sets $\Gg := \{ \ee \in \Sigma^n : \forall \bfu \in C^{\perp}, \Dec_{C^{\perp}}(\bfu + \ee) = \bfu \}$, $\Bb := \Sigma^n \setminus \Gg$, $\GOOD := C^{\perp} \times \Gg$ and $\BAD := (\Sigma^n \times \Sigma^n) \setminus \GOOD$. By construction, $\Dec_{C^{\perp}}(\bfu + \ee) = \bfu$ for all $(\bfu, \ee) \in \GOOD$. Applying the definition of $\widehat{V}$ from \Cref{lemma:dual_fourier}, we see that for all $\bb \in \{0, 1\}^n$ and functions $H$,
\begin{align*}
    \sum_{(\bfu, \ee) \in \BAD} \left|\widehat{V}(\bfu)\widehat{W}^{H, \bb}(\ee)\right|^2 = \sum_{\bfu \in C^{\perp}} \sum_{\ee \in \Bb} \left|\frac{1}{\sqrt{|C^{\perp}|}}\widehat{W}^{H, \bb}(\ee)\right|^2 = \sum_{\ee \in \Bb} \left|\widehat{W}^{H, \bb}(\ee)\right|^2\,. 
\end{align*}
We can apply the same logic to get that
\begin{equation*}
    \sum_{\z \in \Sigma^n}\left|\sum_{(\bfu, \ee) \in \BAD:\bfu+\ee = \z} \widehat{V}(\bfu)\cdot \widehat{W}^{H, \bb}(\ee)\right|^2 = \sum_{\z \in \Sigma^n}\left|\sum_{\substack{\bfu \in C^{\perp}, \ee \in \Bb: \\ \bfu+\ee = \z}} \widehat{V}(\bfu)\cdot \widehat{W}^{H, \bb}(\ee)\right|^2\,. 
\end{equation*}

\begin{definition}
    Let $\Dd_{p, b}$ be the distribution over $\Sigma$ that takes 0 with probability $1-p$ if $b = 0$ (resp. probability $p$ if $b = 1$) and otherwise takes a uniformly random element of $\Sigma \setminus \{0\}$. For any bitstring $\bb \in \{0, 1\}^n$, define the distribution $\Dd_{p, \bb}$ over $\Sigma^n$ to be the Cartesian product of the distributions $\Dd_{p, \bb_i}$.
\end{definition}

\begin{definition}
    Fix any $\FF_q$-linear code $C \subseteq \Sigma^n$, function $\Dec_{C^{\perp}}: \Sigma^n \to \Sigma^n$, set $S \subseteq \{0, 1\}^n$, $p \in [0, 1]$, and real-valued function $\mu: \NN \to \RR$. $(C, \Dec_{C^{\perp}})$ is said to be $(p, \mu, S)$-good if for all $\lambda \in \NN$ and $\bb \in S$, $\Pr_{\ee \gets \Dd_{p, \bb}}[\ee \in \Bb] \leq \mu(\lambda)$.
\end{definition}

\begin{lemma}\label{lemma:condition_one}
    Suppose that $(C, \Dec_{C^{\perp}})$ is $(p, \mu, S)$-good. Then, for all $\lambda \in \NN$ and $\bb \in S$,
    \begin{equation*}    
        \Exp_{H \gets \Bias_{n, p, \Sigma}}\left[\sum_{\ee \in \Bb} \left|\widehat{W}^{H, \bb}(\ee)\right|^2\right] \leq \mu(\lambda)\,.
    \end{equation*}
\end{lemma}
\begin{proof}
    By \Cref{claim:fourier_coeffs}, $\Dd_{p, b_i}(e_i) = \Exp_{H_i}[|\widehat{W}_i^{H_i, b_i}(e_i)|^2]$ for all $e_i \in \Sigma$ where (slightly abusing notation) $\Dd_{p, b_i}(\cdot)$ is the density function of the distribution $\Dd_{p, b_i}$. Moreover, for any $\ee = (e_1, \ldots, e_n) \in \Sigma^n$, string $\bb$ and function $H$, since $W^{H, \bb}(\ee) = \prod_{i = 1}^n W_i^{H_i, b_i}(e_i)$, by \Cref{lemma:fourier_properties}, we have that $\widehat{W}^{H, \bb}(\ee) = \prod_{i = 1}^n \widehat{W}_i^{H_i, b_i}(e_i)$. Thus, $\Dd_{p, \bb}(\ee) = \Exp_{H \gets \Bias_{n, p, \Sigma}}[|\widehat{W}^{H, \bb}(\ee)|^2]$ for all $\ee \in \Sigma^n$. By linearity of expectation, we see that for all $\lambda \in \NN$ and $\bb \in S$,
    \begin{equation*}
        \Exp_{H \gets \Bias_{n, p, \Sigma}}\left[\sum_{\ee \in \Bb} \left|\widehat{W}^{H, \bb}(\ee)\right|^2\right] = \sum_{\ee \in \Bb}\Exp_{H \gets \Bias_{n, p, \Sigma}}\left[\left|\widehat{W}^{H, \bb}(\ee)\right|^2\right] = \sum_{\ee \in \Bb} \Dd_{p, \bb}(\ee) = \Pr_{\ee \gets \Dd_{p, \bb}}[\ee \in \Bb] \leq \mu(\lambda)\,. \qedhere
    \end{equation*}
\end{proof}

We define the function $B : \Sigma^n \to \CC$ to be the inverse Fourier transform of the indicator function $\widehat{B}(\ee) = \Id_{\ee \in \Bb}$.

\begin{claim}\label{claim:small_fourier}
    Suppose that $(C, \Dec_{C^{\perp}})$ is $(p, \mu, S)$-good. Then for all $\lambda \in \NN$, $\bb \in S$, and $\z \in \Sigma^n$,
        \[ \Exp_{H \gets \Bias_{n, p, \Sigma}}\left[\left|(B \star W^{H, \bb})(\z)\right|^2\right] \leq \mu(\lambda). \]
\end{claim}
\begin{proof}
    For $\z_0, \z_1 \in \Sigma^n$, define the permutation $\pi_{\z_0, \z_1}: \Sigma^n \to \Sigma^n$ as $\pi_{\z_0, \z_1}(\z) := \z+\z_0-\z_1$. By \Cref{claim:perms},
    \begin{align*}
        \Exp_{H \gets \Bias_{n, p, \Sigma}}\left[\left|(B \star W^{H, \bb})(\z_0)\right|^2\right] &= \Exp_{H \gets \Bias_{n, p, \Sigma}}\left[\left|\sum_{\x \in \Sigma^n} B(\x) \cdot W^{H, \bb}(\z_0-\x)\right|^2\right] \\
        &= \Exp_{H \gets \Bias_{n, p, \Sigma}}\left[\left|\sum_{\x \in \Sigma^n} B(\x) \cdot W^{H \circ \pi_{\z_0, \z_1}, \bb}(\z_1-\x)\right|^2\right] \\
        &= \Exp_{H \gets \Bias_{n, p, \Sigma}}\left[\left|\sum_{\x \in \Sigma^n} B(\x) \cdot W^{H, \bb}(\z_1-\x)\right|^2\right] \\
        &= \Exp_{H \gets \Bias_{n, p, \Sigma}}\left[\left|(B \star W^{H, \bb})(\z_1)\right|^2\right]\,.
    \end{align*}
    Thus, by \Cref{lemma:condition_one}, we have that for all $\lambda \in \NN$, $\bb \in S$, and $\z \in \Sigma^n$,
    \begin{align*}
        \Exp_{H \gets \Bias_{n, p, \Sigma}}\left[\left|(B \star W^{H, \bb})(\z)\right|^2\right] &= \frac{1}{|\Sigma|^n} \Exp_{H \gets \Bias_{n, p, \Sigma}}\left[\sum_{\z \in \Sigma^n}\left|(B \star W^{H, \bb})(\z)\right|^2\right] \\
        &= \frac{1}{|\Sigma|^n} \Exp_{H \gets \Bias_{n, p, \Sigma}}\left[\sum_{\z \in \Sigma^n}\left||\Sigma|^{n/2}\widehat{B}(\z) \cdot \widehat{W}^{H, \bb}(\z)\right|^2\right] \\
        &= \Exp_{H \gets \Bias_{n, p, \Sigma}}\left[\sum_{\z \in \Sigma^n}\left|\widehat{B}(\z) \cdot \widehat{W}^{H, \bb}(\z)\right|^2\right] \\
        &= \Exp_{H \gets \Bias_{n, p, \Sigma}}\left[\sum_{\z \in \Bb}\left|\widehat{W}^{H, \bb}(\z)\right|^2\right] \leq \mu(\lambda)\,. \qedhere
    \end{align*}
\end{proof}

\begin{claim}\label{claim:fourier_manip}
For any function $H$ and string $\bb \in \{0, 1\}^n$, it holds that
    \[ \sum_{\z \in \Sigma^n}\left|\sum_{\substack{\bfu \in C^{\perp}, \ee \in \Bb: \\ \bfu+\ee = \z}} \widehat{V}(\bfu)\cdot \widehat{W}^{H, \bb}(\ee)\right|^2 = \sum_{\z \in \Sigma^n}\left|(V \cdot (B \star W^{H, \bb}))(\z)\right|^2. \]
\end{claim}
\begin{proof}
    For any $\z \in \Sigma^n$, we use the fact that $V(\x) = 0$ for $\x \notin C^{\perp}$ to show that
    \begin{align*}
        \sum_{\substack{\bfu \in C^{\perp}, \ee \in \Bb: \\ \bfu+\ee = \z}} \widehat{V}(\bfu)\cdot \widehat{W}^{H, \bb}(\ee) = \sum_{\substack{\bfu \in \Sigma^n, \ee \in \Sigma^n: \\ \bfu+\ee = \z}} \widehat{V}(\bfu)\cdot \widehat{B}(\ee) \cdot \widehat{W}^{H, \bb}(\ee) = (\widehat{V} \star (\widehat{B} \cdot \widehat{W}^{H, \bb}))(\z) = (\widehat{V \cdot (B \star W^{H, \bb})})(\z)\,.
    \end{align*}
    The claim then follows from Parseval’s equality.
\end{proof}

\begin{corollary}\label{cor:condition_two}
Suppose that $(C, \Dec_{C^{\perp}})$ is $(p, \mu, S)$-good. Then for all $\lambda \in \NN$ and $\bb \in S$,
    \[ \Exp_{H \gets \Bias_{n, p, \Sigma}}\left[\sum_{\z \in \Sigma^n}\left|\sum_{\substack{\bfu \in C^{\perp}, \ee \in \Bb: \\ \bfu+\ee = \z}} \widehat{V}(\bfu)\cdot \widehat{W}^{H, \bb}(\ee)\right|^2\right] \leq \mu(\lambda). \]
\end{corollary}
\begin{proof}
By \Cref{claim:small_fourier,claim:fourier_manip}, we have that for all $\lambda \in \NN$ and $\bb \in S$,
\begin{align*}
    \Exp_{H \gets \Bias_{n, p, \Sigma}}\left[\sum_{\z \in \Sigma^n}\left|\sum_{\substack{\bfu \in C^{\perp}, \ee \in \Bb: \\ \bfu+\ee = \z}} \widehat{V}(\bfu)\cdot \widehat{W}^{H, \bb}(\ee)\right|^2\right] &= \Exp_{H \gets \Bias_{n, p, \Sigma}}\left[\sum_{\z \in \Sigma^n}\left|(V \cdot (B \star W^{H, \bb}))(\z)\right|^2\right] \\
    &= \Exp_{H \gets \Bias_{n, p, \Sigma}}\left[\sum_{\z \in C} \frac{1}{|C|}\left|(B \star W^{H, \bb})(\z)\right|^2\right] \\
    &= \frac{1}{|C|} \sum_{\z \in C} \Exp_{H \gets \Bias_{n, p, \Sigma}}\left[\left|(B \star W^{H, \bb})(\z)\right|^2\right] \\
    &\leq \frac{1}{|C|} \sum_{\z \in C} \mu(\lambda) = \mu(\lambda)\,. \qedhere
\end{align*}
\end{proof}

\subsection{The Biased Yamakawa-Zhandry Algorithm}\label{subsec:biased_yz}
We are now ready to present our modified algorithm for handling biased oracles.
\begin{definition}[The Yamakawa-Zhandry advice state]\label{def:advice_state}
    For a code $C \subseteq \Sigma^n$ and function $H: [n] \times \Sigma \to \{0, 1\}$, define the sets $S_{i, b} := \{e \in \Sigma: H(i, e) = b\}$ for $(i, b) \in [n] \times \{0, 1\}$. Let $\ket{\phi_{i, b}}$ and $\ket{\psi}$ denote the following states:
    \begin{align*}
        \ket{\phi_{i, b}} = \begin{cases}
            \frac{1}{\sqrt{|S_{i, b}|}} \sum_{e \in S_{i, b}} \ket{e} & \text{if $|S_{i, b}| \neq 0$} \\
            \ket{\bot} & \text{otherwise}
        \end{cases} \quad \text{and} \quad \ket{\psi} = \frac{1}{\sqrt{|C|}} \sum_{\bfu \in C} \ket{\bfu}\,.
    \end{align*}
    We define the \emph{advice state} for $(C, H)$, denoted by $\ket{\adv_{C, H}}$ (or $\ket{\adv_H}$ when the code $C$ is implicit) as follows:
    \begin{align*}
        \ket{\adv_H} := \begin{cases}
            (\bigotimes_{i = 1}^{n} \ket{\phi_{i, 0}} \otimes \ket{\phi_{i, 1}}) \otimes \ket{\psi} & \text{if $\ket{\phi_{i, b}} \neq \ket{\bot}$ for all $(i, b) \in [n] \times \{0, 1\}$,} \\
            \ket{\bot} & \text{otherwise.}
        \end{cases}
    \end{align*}
\end{definition}
\begin{theorem}\label{thm:biased_yz}
    Fix any $\FF_q$-linear code $C \subseteq \Sigma^n = (\FF_q^s)^n$, function $\Dec_{C^{\perp}}: \Sigma^n \to \Sigma^n$, set $S \subseteq \{0, 1\}^n$, $p \in [0, 1/2]$, and function $\mu: \NN \to [0, 1]$ such that $(C, \Dec_{C^{\perp}})$ is $(p, \mu, S)$-good. Let $\BiasedYZ = \BiasedYZ_{\Dec_{C^{\perp}}}$ be the quantum algorithm described in \Cref{fig:yz_alg} and $\ket{\adv_H}$ be the state described in Definition \ref{def:advice_state}. Then, for all $\lambda \in \NN$:
    \begin{enumerate}
        \item $\BiasedYZ$ runs in time $\poly(n, s, \log q, \lambda)+T_{\Dec}$, where $T_{\Dec}$ is the time required to compute $\Dec_{C^{\perp}}$.
        \item For all $C$ and $H$, $\ket{\adv_H}$ is a $O(sn\log q)$-qubit state.
        \item For all strings $\bb \in S$,
        \begin{multline*}\Pr_{H \gets \Bias_{n, p, \Sigma}}\left[\Pr_{\BiasedYZ}[(\bb, \bfv) \in R_{C, H}: \bfv \leftarrow \BiasedYZ(\ket{\adv_H}, \bb)] \geq 1-16\mu(\lambda)^{1/4}-2^{-4\lambda} \right] \\\geq 1-(2n(1-p)^{|\Sigma|}+2\mu(\lambda)^{1/2})\,. \end{multline*}
    \end{enumerate}
    \begin{figure}[ht]
    \begin{mdframed}
        \textbf{$\BiasedYZ_{\Dec_{C^{\perp}}}(\ket{\adv_H}, \bb)$:}
        \begin{enumerate}
            \item If $\ket{\adv_H} \neq \ket{\bot}$, construct the state $(\bigotimes_{i = 1}^n \ket*{\phi_{i, \bb_i}})_{\reg{A}} \otimes \ket{\psi}_{\reg{B}}$ by re-arranging $\ket{\adv_H}$; else, return $\bot$.
            \item Apply $(\mathsf{QFT}_{\Sigma}^{\otimes n})_{\reg{A}} \otimes (\mathsf{QFT}_{\Sigma}^{\otimes n})_{\reg{B}}$.
            \item Controlled on register $\reg{B}$, add the value of register $\reg{B}$ to register $\reg{A}$.
            \item Apply $\Dec_{C^{\perp}}$ to uncompute register $\reg{B}$ given the value of register $\reg{A}$.  
            \item Apply $(\mathsf{QFT}_{\Sigma}^{-1, \otimes n})_{\reg{A}} \otimes \id_{\reg{B}}$ and measure register $\reg{A}$ to get an outcome $\bfv \in \Sigma^n$. Output $\bfv$.
        \end{enumerate}
    \end{mdframed}
    \caption{Biased Yamakawa-Zhandry Algorithm $\BiasedYZ_{\Dec_{C^{\perp}}}(\ket{\adv_H}, \bb)$}
    \label{fig:yz_alg}
    \end{figure}
\end{theorem}
\begin{proof}
    Throughout this proof, we assume that $2\mu(\lambda)^{1/4} < 1$, since the theorem holds trivially otherwise. Upon inspection, we see that $\ket{\adv_H}$ is a $O(2n \log(|\Sigma|) + \log(|\Sigma|^n)) = O(sn\log q)$-qubit state.
    
    We now analyze the runtime of $\BiasedYZ$. First, observe that step 1 runs in time $O(sn\log q)$. To implement steps 2 and 5, note that $\QFT^{\otimes n}_{\Sigma}$ (and its inverse) can be implemented with total error at most $2^{-4\lambda}/4$ in $\poly(n, s, \log q, \log 2^{4\lambda}) = \poly(n, s, \log q, \lambda)$ time. Finally, step 3 consists of adding in $\Sigma$, which can be done in $\poly(n, s, \log q)$ time, and step 4 takes time $T_{\Dec}$ by definition. In total, we have that $\BiasedYZ$ runs in time $O(sn \log q)+\poly(n, s, \log q, \lambda)+\poly(n, s, \log q)+T_{\Dec} = \poly(n, s, \log q, \lambda)+T_{\Dec}$.

    We finish by analyzing the correctness of $\BiasedYZ$. For fixed $(i, b) \in [n] \times \{0, 1\}$, $\ket{\phi_{i, b}} \neq \ket{\bot}$ with probability at least $1-(1-p)^{|\Sigma|}$ (as long as at least one symbol hashes to $b$ under $H(i, \cdot)$). By a union bound, $\ket{\adv_H} \neq \ket{\bot}$ with probability at least $1-2n(1-p)^{|\Sigma|}$. Thus, using the definition of $V$ and $W$ from Section \ref{subsec:technical_lemmas}, we will assume for the remainder of the proof that $\ket{\adv_H} = \sum_{\bfu, \ee \in \Sigma^n} V(\bfu) W(\ee) \ket{\bfu} \ket{\ee}$. Then, by \Cref{lemma:condition_one} and \Cref{cor:condition_two}, we have that for all $\bb \in S$,
    \begin{equation*}
        \Exp_{H \gets \Bias_{n, p, \Sigma}}\left[\sum_{(\bfu, \ee) \in \BAD} \left|\widehat{V}(\bfu)\cdot \widehat{W}^{H, \bb}(\ee)\right|^2\right] = \Exp_{H \gets \Bias_{n, p, \Sigma}}\left[\sum_{\ee \in \Bb} \left|\widehat{W}^{H, \bb}(\ee)\right|^2\right] \leq \mu(\lambda)\, 
    \end{equation*}
    and 
    \begin{align*} 
        \Exp_{H \gets \Bias_{n, p, \Sigma}}\left[\sum_{\z \in \Sigma^n}\left|\sum_{\substack{(\bfu, \ee) \in \BAD: \\ \bfu+\ee = \z}} \widehat{V}(\bfu)\cdot \widehat{W}^{H, \bb}(\ee)\right|^2\right] = \Exp_{H \gets \Bias_{n, p, \Sigma}}\left[\sum_{\z \in \Sigma^n}\left|\sum_{\substack{\bfu \in C^{\perp}, \ee \in \Bb: \\ \bfu+\ee = \z}} \widehat{V}(\bfu)\cdot \widehat{W}^{H, \bb}(\ee)\right|^2\right] \leq \mu(\lambda)\,.
    \end{align*}
    Fixing $\bb \in S$, Markov's inequality and the union bound implies that $H$ satisfies both
    \begin{align*}
        \sum_{(\bfu, \ee) \in \BAD} \left|\widehat{V}(\bfu)\cdot \widehat{W}^{H, \bb}(\ee)\right|^2 \leq \mu(\lambda)^{1/2} \quad \text{and} \quad \sum_{\z \in \Sigma^n}\left|\sum_{(\bfu, \ee) \in \BAD:\bfu+\ee = \z} \widehat{V}(\bfu)\cdot \widehat{W}^{H, \bb}(\ee)\right|^2 \leq \mu(\lambda)^{1/2}\,
    \end{align*}
    with probability at least $1-2\mu(\lambda)^{1/2}$ over $\Bias_{n, p, \Sigma}$. For these $H$, by \Cref{lemma:yz_regev},
    \begin{multline*}
        \bigg\lVert\left(\QFT^{-1, \otimes n}_{\Sigma} \otimes \mathrm{id}\right) \cdot U_{\Dec_{C_{\lambda}^{\perp}}} \cdot U_{\mathsf{add}} \cdot \left(\QFT^{\otimes n}_{\Sigma} \otimes \QFT^{\otimes n}_{\Sigma}\right) \bigotimes_{i = 1}^n \ket*{\phi_{i, \bb_i}} \otimes \ket{\psi}
        \\- |\Sigma|^{n/2} \sum_{\z \in \Sigma^n} (V(\z) \cdot W^{H, \bb}(\z)) \ket{\z} \ket{0}\bigg\rVert \leq 2\mu(\lambda)^{1/4}\,.
    \end{multline*}
    Since $\ket{\adv_H} \neq \ket{\bot}$ and $2\mu(\lambda)^{1/4} < 1$, the state $\ket{\mathsf{tgt}} = |\Sigma|^{n/2} \sum_{\z \in \Sigma^n} (V(\z) \cdot W^{H, \bb}(\z)) \ket{\z}$ is nonzero. But measuring (the normalization of) $\ket{\mathsf{tgt}}$ in the standard basis always produces vectors $\bfv$ such that $(\bb, \bfv) \in R_{C, H}$. Thus, by \Cref{cor:measuring_distance}, the output of $\BiasedYZ$ will be a vector $\bfv$ such that $(\bb, \bfv) \in R_{C_{\lambda}, H}$ with probability at least $1-16\mu(\lambda)^{1/4}-2^{-4\lambda}$ (where the $2^{-4\lambda}$ term comes from approximating the QFT).
    
    We conclude that with probability $1-(2n(1-p)^{|\Sigma|}+2\mu(\lambda)^{1/2})$ over $\Bias_{n, p, \Sigma}$, $\BiasedYZ$ succeeds with probability at least $1-16\mu(\lambda)^{1/4}-2^{-4\lambda}$ for any given $\bb \in S$, as desired.
\end{proof}

Finally, we show that our choice of code satisfies the required conditions of \Cref{thm:biased_yz}.

\begin{corollary}
    For each security parameter $\lambda \in \NN$, define the code $C_{\lambda} = \Mult_{s, \FF_q, k}$, where $\lambda^5 < q \leq 2\lambda^5$ is a prime, $k = \lambda^3$, and $s = \lambda$. Then, there exists an efficient/uniform quantum algorithm $\BiasYZ$ and a family of $\poly(\lambda)$-qubit states $\{\ket{\adv_H}\}_H$ such that the following holds for sufficiently large $\lambda$: for all strings $x \in \{0, 1\}^{\lambda}$,
        \[ \Pr_{H \gets \Bias_{q, 1/\lambda^4, \FF_q^s}}\left[\Pr[(x \| 0^{q-\lambda}, \bfv) \in R_{C_{\lambda}, H}: \bfv \leftarrow \BiasYZ(\ket{\adv_H}, x)] \geq 1-2^{-\lambda} \right] \geq 1-2^{-2\lambda}. \]
\end{corollary}
\begin{proof}
We begin by arguing that there exists a deterministic algorithm $\Dec_{C_{\lambda}^{\perp}}$ and function $\mu$ that for sufficiently large $\lambda$ satisfies $\mu(\lambda) \leq 2^{-8\lambda}$, such that $T_{\Dec} = \poly(\lambda)$ and $(C_{\lambda}, \Dec_{C_{\lambda}^{\perp}})$ is $\left(\frac{1}{\lambda^4}, \mu, \{0, 1\}^{\lambda} \times 0^{q-\lambda}\right)$-good. Define the subset $\Gg := \left\{\ee \in \Sigma^q: \hw(\ee_{[\lambda+1:q]}) \leq \frac{11}{\lambda^4}(q-\lambda)\right\}$. By the Chernoff bound (\Cref{lemma:chernoff}), for any $x \in \{0, 1\}^{\lambda}$,
    \begin{align*} 
        \Pr_{\ee \gets \Dd_{1/\lambda^4, x \| 0^{q-\lambda}}}[\ee \notin \Gg] = \Pr_{\ee \gets \Dd_{1/\lambda^4, x \| 0^{q-\lambda}}}\left[\hw(\ee_{[\lambda+1:q]}) > 11 \cdot \frac{(q-\lambda)}{\lambda^4}\right] \leq e^{-\frac{100\lambda}{12}+10} \leq 2^{-8\lambda}\,,
    \end{align*}
assuming $\lambda \geq 10$, so it suffices to consider decoding errors $\ee \in \Gg$. \Cref{thm:mult_dual} implies that for all $\ee \in \Gg$,
\begin{align*}
    \hw(\ee) = \hw(\ee_{[1:\lambda]}) + \hw(\ee_{[\lambda+1:q]}) < \lambda + \frac{11}{\lambda^4}(q-\lambda) \leq 25\lambda \leq \lambda^2/2 \leq \dist(C_{\lambda}^{\perp})/2\,,
\end{align*}
assuming $\lambda \geq 50$. It therefore suffices to \emph{uniquely} decode $C_{\lambda}^{\perp}$ in a deterministic and efficient manner.  We set $\mu$ to be the maximum probability (across all $x$) that unique decoding for $C_{\lambda}^{\perp}$ fails with error distribution $\mathcal{D}_{1/\lambda^4, x\|0^{q - \lambda}}$; by our previous argument, $\mu(\lambda) \leq 2^{-8\lambda}$ as long as $\lambda \geq 50$.
    
By \Cref{thm:mult_duality}, we know that $C_{\lambda}^{\perp} = \mathsf{GM}_{s, \FF_q}(U_1, \ldots, U_q; 1, \ldots, q; sq-k)$. From \Cref{lemma:hermite_interpolation}, it is easy to see that $A_{i, j}(X)$ can be computed in $\poly(s, q) = \poly(\lambda)$ time, and thus $a_{i, j}$ (and consequently $U_i$ and $U_i^{-1}$) can be computed in $\poly(\lambda)$ time. Finally, it remains to efficiently uniquely decode $\Mult_{s, \FF_q, sq-k}$, which we can do deterministically in $\poly(s, q) = \poly(\lambda)$ time \cite{Nie01, KSY14, Kop15}.\footnote{In fact, a simple extension of the Berlekamp-Welch algorithm \cite{WB86} gives efficient unique decoding for univariate multiplicity codes.}

Applying \Cref{thm:biased_yz} gives a family of $\poly(\lambda)$-qubit states $\{\ket{\adv_H}\}_H$ and a $\poly(\lambda)$-time uniform algorithm $\Bb$ such that for $\lambda \geq 50$ and all strings $x \in \{0, 1\}^{\lambda}$,
\begin{align*}
    \Pr_{H \gets \Bias_{q, 1/\lambda^4, \Sigma}}\left[\Pr[(x \| 0^{q-\lambda}, \bfv) \in R_{C_{\lambda}, H}: \bfv \leftarrow \Bb(\ket{\adv_H}, x \| 0^{q-\lambda})] \geq 1-2^{-2\lambda+4}-2^{-4\lambda} \right] &\geq 1-2q(1-p)^{|\Sigma|}-2^{-4\lambda+1} \\
    \implies \Pr_{H \gets \Bias_{q, 1/\lambda^4, \FF_q^s}}\left[\Pr[(x \| 0^{q-\lambda}, \bfv) \in R_{C_{\lambda}, H}: \bfv \leftarrow \Bb(\ket{\adv_H}, x \| 0^{q-\lambda})] \geq 1-2^{-\lambda} \right] &\geq 1-2^{-2\lambda}.
\end{align*}
The algorithm $\BiasYZ$ simply runs $\Bb(\ket{\adv_H}, x \| 0^{q-\lambda})$ given advice $\ket{\adv_H}$ and input $x \in \{0, 1\}^{\lambda}$.
\end{proof}

We therefore have the following corollary by a simple union bound over all $x \in \{0, 1\}^{\lambda}$.
\begin{corollary}\label{cor:worst_case_yz}
    For each security parameter $\lambda \in \NN$, define the code $C_{\lambda} = \Mult_{s, \FF_q, k}$, where $\lambda^5 < q \leq 2\lambda^5$ is a prime, $k = \lambda^3$, and $s = \lambda$. Then, there exists an efficient/uniform quantum algorithm $\BiasYZ$ and a family of $\poly(\lambda)$-qubit states $\{\ket{\adv_H}\}_H$ such that for sufficiently large $\lambda$,
    \begin{equation*}
        \Pr_{H \gets \Bias_{q, 1/\lambda^4, \FF_q^s}}\left[\forall x \in \{0, 1\}^{\lambda}, \Pr[(x \| 0^{q-\lambda}, \bfv) \in R_{C_{\lambda}, H}: \bfv \leftarrow \BiasYZ(\ket{\adv_H}, x)] \geq 1-2^{-\lambda} \right] \geq 1-2^{-\lambda}\,.
    \end{equation*}
\end{corollary}

\section{\texorpdfstring{Separating $\QMA$ from $\QCMA$}{Separating QMA from QCMA}}
For the rest of the paper, we fix a code family $C_{\lambda} := \Mult_{s, \FF_q, k}$ for $\lambda^5 < q \leq 2\lambda^5$ a prime,\footnote{To be concrete, we can take $q$ to be the smallest prime larger than $\lambda^5$ (which is always at most $2\lambda^5$).} $k = \lambda^3$, and $s = \lambda$. For any subset $E \subseteq \{0, 1\}^{\lambda} \times \Sigma^q$ and function $H: [q] \times \Sigma \to \{0, 1\}$, we define the oracle
    \[ O[H, E](x, \bfv) = \begin{cases} 1 & \text{if } (x \| 0^{q-\lambda}, \bfv) \in R_{C_{\lambda}, H} \land (x, \bfv) \in E, \\ 0 & \text{otherwise.} \end{cases} \]
Our proofs in this section are fairly standard and follow \cite{BHNZ25}, but we include them for completeness.

We begin by defining an oracle-input problem based on the code intersection subset size checking problem.  Note that our NO instances are slightly different than those considered in \cite{LLPY23, BDK24}, which will affect the $\QCMA$ lower bound, but not the $\QMA$ containment.
\begin{definition}[YES and NO instances of the code intersection subset size problem]\label{def:yes_no_instances}
    Let $\Good$ denote the set of functions $H: [q] \times \FF_q^s \to \{0, 1\}$ such that the algorithm in \Cref{cor:worst_case_yz} succeeds on all $x \in \{0, 1\}^{\lambda}$ with probability at least $2/3$. Our oracle-input separation between $\QMA$ and $\QCMA$ involves distinguishing between $O[H, E]$ with $H$ and $E$ being the following:
    \begin{enumerate}
        \item YES instances: $H \in \mathsf{Good}$ and $E = \{0, 1\}^{\lambda} \times \Sigma^q$.
        \item NO instances: $H \in \Good$ and subsets $E \subseteq \{0, 1\}^{\lambda} \times \Sigma^q$ such that $|E| \leq 2^{\lambda}/3$.
    \end{enumerate}
\end{definition}
\begin{remark}\label{remark:bin_input}
    One can also embed $O[H, E]: \{0, 1\}^{\lambda} \times \Sigma^q \to \{0, 1\}$ into a binary input domain oracle with $\lambda + q\log |\Sigma| \leq \lambda^7$-bit inputs, in such a way that for sufficiently large $\lambda \geq \lambda_0$ (where $\lambda_0 \in \NN$ is some constant), there is at most one security parameter associated with each input length.
\end{remark}

\subsection{\texorpdfstring{The $\QMA$ Proof System}{The QMA Proof System}}\label{subsec:qma}
We first show that there is a $\QMA$ proof system that distinguishes between YES and NO instances of the code intersection subset size checking problem, as defined in \Cref{def:yes_no_instances}.

\begin{lemma}[A $\QMA$ proof system]\label{lemma:qma_easiness}
    There exists a polynomial-time uniform quantum query algorithm $V$ which makes one query to the oracle $O[H, E]$, such that for sufficiently large $\lambda$, the following holds:
    \begin{enumerate}
        \item \textbf{Completeness.} For all $H \in \Good$, when $E = \{0, 1\}^{\lambda} \times \Sigma^q$, there exists a $\poly(\lambda)$-qubit state $\ket{\adv_H}$ such that 
            \[ \Pr[V^{O[H,E]}(\ket{\adv_H}) = 1] \geq \frac{2}{3}. \]
        \item \textbf{Soundness.} For all $H$, sets $E \subseteq \{0, 1\}^{\lambda} \times \Sigma^q$ where $|E| \leq 2^{\lambda}/3$, and quantum states $\ket{\adv^{*}_H}$,
            \[ \Pr[V^{O[H, E]}(\ket{\adv^{*}_H}) = 1] \leq \frac{1}{3}. \]
    \end{enumerate}
\end{lemma}
\begin{proof}
    By definition, for $H \in \Good$ as defined in \Cref{cor:worst_case_yz}, there exists a $\poly(\lambda)$-qubit state $\ket{\adv_H}$ and efficient algorithm $\BiasYZ$ such that for any $x \in \{0, 1\}^{\lambda}$,
        \[ \Pr_{\BiasYZ}[(x \| 0^{q-\lambda}, \bfv) \in R_{C_{\lambda}, H}: \bfv \leftarrow \BiasYZ(\ket{\adv_H}, x)] \geq 1-2^{-\lambda} \geq \frac{2}{3}. \]
    The verifier $V$ operates as follows: it samples a uniformly random $x \in \{0, 1\}^{\lambda}$ and runs $\bfv \gets \BiasYZ(\ket{\adv_H}, x)$. Finally, $V$ queries $O$ at $(x, \bfv)$ and returns the output of $O$. The efficiency and uniformity of $V$ follows from the efficiency and uniformity of $\BiasYZ$ combined with the fact that $V$ makes one oracle query. 
    
    We now argue completeness and soundness. If $E=\{0, 1\}^{\lambda} \times \Sigma^q$, then 
    \begin{align*}
        \Pr[V^{O[H,E]}(\ket{\adv_H}) = 1] = \Pr_{\BiasYZ, x}[(x \| 0^{q-\lambda}, \bfv) \in R_{C_{\lambda}, H}: \bfv \leftarrow \BiasYZ(\ket{\adv_H}, x)] \geq \frac{2}{3}\,.
    \end{align*}
    On the other hand, note that $V$ will always output 0 if there does not exist a $\bfv$ such that $(x, \bfv) \in E$. Thus, if $|E| \leq 2^\lambda/3$, then for all $H$ and quantum states $\ket{\adv^{*}_H}$,
    \begin{equation*}
        \Pr[V^{O[H, E]}(\ket{\adv^{*}_H}) = 1] \leq \Pr_{x \gets \{0, 1\}^{\lambda}}[\exists \bfv: (x, \bfv) \in E] \leq \frac{|E|}{2^{\lambda}} \leq \frac{1}{3}\,. \qedhere
    \end{equation*}
\end{proof}

\subsection{\texorpdfstring{Non-Existence of $\QCMA$ Proof Systems}{Non-Existence of QCMA Proof Systems}}\label{subsec:no_qcma}
We begin by showing that any $\QCMA$ verifier/algorithm can be turned into a very good hash value guesser.
\begin{lemma}[Good guessers from $\QCMA$ algorithms]\label{lemma:qcma_guesser}
    Assume there exists a $\QCMA$ algorithm $\Aa$ such that for instances of size $\lambda$, $\Aa$ takes a $t(\lambda)$-bit witness and makes $Q(\lambda)$ oracle queries. Then for all $\ell \leq 2^{\lambda / 3}$, there exists a algorithm $\mathsf{Guesser}$ which makes \emph{no} queries such that for any $H \in \Good$,
        \[ \Pr\left[\begin{aligned} &\forall i \neq j, (x_i, \bfv_i) \neq (x_j, \bfv_j) \\ &\land (x_i \| 0^{q-\lambda}, \bfv_i) \in R_{C_{\lambda}, H} \end{aligned} : \{(x_i, \bfv_i)\}_{i = 1}^{\ell} \gets \mathsf{Guesser}(1^{\ell})\right] \geq 2^{-t(\lambda)} \cdot \left(\frac{1}{144Q(\lambda)^2}\right)^\ell. \]
\end{lemma}
\begin{proof}
    The algorithm $\Aa^{\Oo[H, E]}$ can be thought of starting from a state $\ket{w, 0}$ and applying a sequence of unitaries $V_0, \ldots, V_Q$ interlaced with queries to $\Oo[H, E]$ before measuring the first qubit in the standard basis. The state of the algorithm right before its final measurement is then given by
    \begin{equation*}
        V_Q\cdot \Oo[H, E]\cdot V_{Q-1}\cdot \Oo[H, E] \ldots \Oo[H, E]\cdot V_0 \ket{w, 0}\,.
    \end{equation*} 
    Let $\mathsf{Guesser}$ be the algorithm described in \Cref{fig:guesser} which outputs $\ell$ tuples of codewords and hash values.
    \begin{figure}[ht]
    \begin{mdframed}
        \textbf{$\mathsf{Guesser}(1^{\ell})$:}
        \begin{enumerate}
            \item Sample a random $w \in \{0, 1\}^t$ and initialize $\Delta_0 = \emptyset$.
            \item For $i \in [\ell]$:
            \begin{enumerate}
                \item Sample $j \gets \{0, \ldots, Q-1\}$ uniformly randomly.
                \item Compute the state $V_j\Oo_{\Delta_{i-1}}V_{j-1}\ldots\Oo_{\Delta_{i-1}}V_0\ket{w, 0}$, where $\Oo_{\Delta_{i-1}}$ is the oracle unitary defined by
                    \[ \ket{x, \bfv} \ket{y} \ket{z} \mapsto \ket{x, \bfv} \ket{y \oplus f_{\Delta_{i-1}}(x, \bfv)} \ket{z}, \text{ where } f_{\Delta_{i-1}}(x, \bfv) := \begin{cases}
                        1 & \text{if $(x, \bfv) \in \Delta_{i-1}$,} \\
                        0 & \text{otherwise.}
                    \end{cases} \]
                \item Measure the first register in the standard basis for output $(x, \bfv)$ and update $\Delta_i := \Delta_{i-1} \cup \{(x, \bfv)\}$.
            \end{enumerate}
            \item Output $\{(x_{[1:\lambda]}, \bfv): (x, \bfv) \in \Delta_{\ell}\}$.
        \end{enumerate}
    \end{mdframed}
    \caption{The Hash Value Guesser, given a successful $\QCMA$ verifier for the $\mathsf{CISS}$ problem.}
    \label{fig:guesser}
    \end{figure}
    
    Let $G$ be the event that the witness $w$ sampled is a good witness for $H$, and let $E_i$ be the event that the $i$'th round of $\mathsf{Guesser}$ appends a tuple $(x \| 0^{q-\lambda}, \bfv) \notin \Delta_{i-1}$ such that $(x \| 0^{q-\lambda}, \bfv) \in R_{C_{\lambda}, H}$.
    \begin{claim}
        $\forall i \in [\ell]$, $\Pr[E_i | E_{i-1} \land \ldots \land E_1 \land G] \geq \frac{1}{144Q(\lambda)^2}$.
    \end{claim}
    \begin{proof}
        Fix an index $1 \leq i \leq \ell$. Observe that conditioned on $E_1 \land \ldots \land E_{i-1}$ occurring, this means that $\Delta_{i-1}$ consists of $i-1$ distinct tuples $\{(x_j \| 0^{q-\lambda}, \bfv_j)\}_{j = 1}^{i-1}$ such that $(x_j \| 0^{q-\lambda}, \bfv_j) \in R_{C_{\lambda}, H}$ for all $1 \leq j \leq i-1$.
        
        Since $\Oo_{\Delta_{i-1}} = \Oo[H, \Delta_{i-1}]$ corresponds to a NO instance (as $|\Delta_{i-1}| = i-1 < \ell \leq 2^{\lambda}/3$) while $\Oo[H, \{0, 1\}^{\lambda} \times \Sigma^q]$ corresponds to a YES instance, by the completeness and soundness of the $\QCMA$ algorithm $\Aa$ and the fact that the event $G$ implies we have a good witness, \Cref{thm:bbbv} implies that the query mass on the inputs where the two oracles differ must be at least $((2/3-1/3)/4)^2/Q = 1/144Q$. 
        
        As $\Oo_{\Delta_{i-1}}$ and $\Oo[H, \{0, 1\}^{\lambda} \times \Sigma^q]$ differ precisely on inputs $(x \| 0^{q-\lambda}, \bfv) \notin \Delta_{i-1}$ where $(x \| 0^{q-\lambda}, \bfv) \in R_{C_{\lambda}, H}$, it follows that measuring a random query of $\Aa$ produces a good tuple $(x \| 0^{q-\lambda}, \bfv)$ with probability at least $\frac{1}{Q} \cdot \frac{1}{144Q} = \frac{1}{144Q^2}$.
    \end{proof}
    Observing that $\Pr[G] \geq 2^{-t(\lambda)}$ as there is always at least one good witness for any $H$, we conclude that
    \begin{equation*}
        \Pr[E_1 \land \ldots \land E_{\ell}] \geq \Pr[G] \cdot \prod_{i = 1}^{\ell} \Pr[E_i | E_{i-1} \land \ldots \land E_1 \land G] \geq 2^{-t(\lambda)} \cdot \left(\frac{1}{144Q(\lambda)^2}\right)^{\ell}. \qedhere
    \end{equation*}
\end{proof}
Separately, we can show the following \emph{upper bound} on the success probability of $\mathsf{Guesser}$. The bound follows from the fact that $\mathsf{Guesser}$ is not making any queries to the oracle, and thus knows nothing about $H$.
\begin{lemma}[Guessing probability upper bound]\label{lemma:guesser_hardness}
For sufficiently large $\lambda$, the following holds: fix any algorithm $\mathsf{Guesser}$ which makes no oracle queries. Then, for all $\ell \leq 2^{\lambda}$,
    \[ \Pr_{H \gets \Bias_{q, 1/\lambda^4, \FF_q^s}}\left[\begin{aligned} &\forall i \neq j, (x_i, \bfv_i) \neq (x_j, \bfv_j) \\ &\land (x_i \| 0^{q-\lambda}, \bfv_i) \in R_{C_{\lambda}, H} \end{aligned}: \{(x_i, \bfv_i)\}_{i = 1}^{\ell} \gets \mathsf{Guesser}(1^{\ell}) \right] \leq \left(1-\frac{1}{\lambda^4}\right)^{\lambda^5 \ell/2}\,. \]
\end{lemma}
\begin{proof}
    Consider any output $\{(x_1, \bfv_1), \ldots, (x_{\ell}, \bfv_{\ell})\}$ of $\mathsf{Guesser}$. First, observe that if $(x_i, \bfv_i)$ are distinct and $H(\bfv_i) = x_i \| 0^{q-\lambda}$, then we must have distinct $\bfv_i$ or else $\mathsf{Guesser}$ will fail. Since we require $\bfv_i \in C_{\lambda}$ for all $i \in [\ell]$, we now consider the sets $S_j := \{ \sigma \in \Sigma \mid \exists i \in [\ell]: (\bfv_i)_j = \sigma \}$ for all $j \in [q]$.
    
    By definition, for all $i \in [\ell]$ and $j \in [q]$, $(\bfv_i)_j \in S_j$, so if we think of $\{S_j\}_{j = 1}^q$ as input lists, the output list for the code $C_{\lambda}$ must contain $\bfv_i \in [\ell]$ and thus $|\{ \bfv \in C_{\lambda} : \forall j \in [q], \bfv_j \in S_j\}| \geq \ell$.
    
    \Cref{cor:mult_list_recovery} thus implies that $\frac{1}{q} \sum_j |S_j| \geq \ell/2$ if $\lambda$ is sufficiently large. In order for $\mathsf{Guesser}$ to succeed, it must correctly guess the output of $H$ on all symbols in $\bigcup_{i = 1}^q S_i$ which contain at least $\sum_j |S_j| \geq \frac{q\ell}{2} = \lambda^5 \ell/2$ distinct points. Since we sample $H \gets \Bias_{q, 1/\lambda^4, \FF_q^s}$, this occurs with probability at most $\left(1-\frac{1}{\lambda^4}\right)^{\lambda^5 \ell/2}$.
\end{proof}

We can combine the upper bound and lower bound to conclude that any $\QCMA$ algorithm $\Aa$ for the code intersection subset size ($\mathsf{CISS}$) problem must misclassify some YES or NO instance.

\begin{lemma}\label{lemma:qcma_hardness}
    For all constants $a > 0$ and functions $Q(\lambda)$, $t(\lambda)$ that satisfy $Q(\lambda) \leq a\lambda^a$, $t(\lambda) \leq a\lambda^a$ for sufficiently large $\lambda$. Then for sufficiently large $\lambda$, for all quantum query algorithms $\Aa$ which take a classical witness of length $t(\lambda)$ and make $Q(\lambda)$ queries to the oracle $O[H, E]$ of size $\lambda$, there exists an oracle $O[H^{*}, E^{*}]$ of size $\lambda$ such that $H^{*} \in \Good$ and
    \begin{enumerate}
        \item either $E^{*} = \{0, 1\}^{\lambda} \times \Sigma^q$, but for all witnesses $w$ of length $t(\lambda)$,
            \[ \Pr[\Aa^{O[H^{*}, E^{*}]}(w) = 1] < \frac{2}{3}\,, \]
        \item or $E^{*} \subseteq F^{*} \times \Sigma^q \subseteq \{0, 1\}^{\lambda} \times \Sigma^q$ where $|F^{*}| \leq 2^{\lambda}/3$, but there exists a witness $\tilde{w}$ of length $t(\lambda)$ such that
            \[ \Pr[\Aa^{O[H^{*}, E^{*}]}(\tilde{w}) = 1] > \frac{1}{3}\,. \]
    \end{enumerate}
\end{lemma}
\begin{proof}
Suppose for the sake of contradiction that $\Aa$ properly classifies all YES and NO instances. Setting $\ell = t \ll 2^{\lambda}/3$, we can apply \Cref{cor:worst_case_yz} and \Cref{lemma:qcma_guesser}, yielding a guesser $\mathsf{Guesser}$ where
    \begin{align*}
        &\Pr_{H \gets \Bias_{q, 1/\lambda^4, \FF_q^s}}\left[\begin{aligned} &\forall i \neq j, (x_i, \bfv_i) \neq (x_j, \bfv_j) \\ &\land (x_i \| 0^{q-\lambda}, \bfv_i) \in R_{C_{\lambda}, H} \end{aligned}: \{(x_i, \bfv_i)\}_{i = 1}^t \gets \mathsf{Guesser}(1^t)\right] \\
        \geq \, &\Pr_{H \gets \Bias_{q, 1/\lambda^4, \FF_q^s}}\left[\begin{aligned} &\forall i \neq j, (x_i, \bfv_i) \neq (x_j, \bfv_j) \\ &\land (x_i \| 0^{q-\lambda}, \bfv_i) \in R_{C_{\lambda}, H} \end{aligned} : \{(x_i, \bfv_i)\}_{i = 1}^t \gets \mathsf{Guesser}(1^t) \middle| H \in \Good \right] \cdot \Pr_{H \gets \Bias_{q, 1/\lambda^4, \FF_q^s}}[H \in \Good] \\
        \geq \, &2^{-t(\lambda)} \cdot \left(\frac{1}{144Q(\lambda)^2}\right)^{t(\lambda)} \cdot (1-2^{-\lambda}) \geq \left(\frac{1}{576Q(\lambda)^2}\right)^{t(\lambda)}.
    \end{align*}
    But this is impossible as \Cref{lemma:guesser_hardness} implies that the success probability of $\mathsf{Guesser}$ is at most
        \begin{equation*}
            \left(1-\frac{1}{\lambda^4}\right)^{\lambda^5 t(\lambda)/2} \leq (e^{-\lambda/2})^{t(\lambda)} \ll \left(\frac{1}{576Q(\lambda)^2}\right)^{t(\lambda)}. \qedhere
        \end{equation*}
\end{proof}
A straightforward diagonalization argument thus gives us our desired separation (see \Cref{app:diag} for details).
\begin{theorem}\label{thm:proof_separation}
    There exists a classical oracle $\Oo: \{0, 1\}^{*} \to \{0, 1\}$ such that $\QMA^{\Oo} \cap \AM^{\Oo} \not\subseteq \QCMA^{\Oo}$.\footnote{We note that this separation can be easily strengthened to $\QMA^{\Oo} \cap \SBP^{\Oo} \not\subseteq \QCMA^{\Oo}$, as with most set-approximation-flavored oracles.}
\end{theorem}

\section{\texorpdfstring{Separating $\BQP/\qpoly$ from $\BQP/\poly$}{Separating BQP/qpoly vs. BQP/poly}}
We begin by proving the main technical result of this section, which is a search-like separation between $\BQP/\qpoly$ and $\BQP/\poly$. In particular, we prove that given as input $x \in \{0, 1\}^{\lambda}$ and oracle access to $O[H, \{0, 1\}^{\lambda} \times \Sigma^q]$ for $H$ sampled from $\Bias_{q, 1/\lambda^4, \FF_q^s}$, the problem of finding $\bfv$ such that $(x \| 0^{q-\lambda}, \bfv) \in R_{C_{\lambda}, H}$ is in $\FBQP/\qpoly$ but not $\FBQP/\poly$ (on average).
\begin{lemma}\label{lemma:bqp_sampling}
    For all security parameters $\lambda \in \NN$, let $\lambda^5 < q < 2\lambda^5$ be a prime, $k = \lambda^3$, and $s = \lambda$, and define the code $C_{\lambda} = \Mult_{s, \FF_q, k}$. In addition, for all $\lambda$, define the set $E_\lambda := \{0, 1\}^{\lambda} \times \Sigma^q$, where $\Sigma = \FF_q^s$. Then the following hold:
    \begin{enumerate}
        \item There is an polynomial-time uniform quantum query algorithm $\Aa$, such that for all oracles $O$ there exists a family of $\poly(\lambda)$-qubit quantum advice states (depending only on the oracle) $\{\ket{z_O}\}_{O}$ such that
            \[ \Pr_{H \gets \Bias_{q, 1/\lambda^4, \FF_q^s}}[\forall x \in \{0, 1\}^{\lambda}, \Pr[(x \| 0^{q-\lambda}, \bfv) \in R_{C_{\lambda}, H}: \bfv \gets \Aa^{O[H, E_\lambda]}(x, \ket{z_O})] \geq 1-\negl(\lambda)] \geq 1-\negl(\lambda). \]
        \item For all unbounded-time quantum algorithms $\Bb$ that make $Q(\lambda) = \poly(\lambda)$ oracle queries to $O$, and all families of $t(\lambda) = \poly(\lambda)$-bit classical advice strings (depending only on the oracle) $\{z_O\}_{O}$,
            \[ \Pr_{H \gets \Bias_{q, 1/\lambda^4, \FF_q^s}, x \gets \{0, 1\}^{\lambda}}[(x \| 0^{q-\lambda}, \bfv) \in R_{C_{\lambda}, H}: \bfv \gets \Bb^{O[H, E_{\lambda}]}(x, z_O)] \leq \negl(\lambda). \]
    \end{enumerate}
\end{lemma}
\begin{proof}
Throughout this proof we will assume that $\lambda$ is sufficiently large and argue with respect to asymptotics. To prove Item 1, we simply set the advice as $\ket{z_O} = \ket{\mathsf{adv}_H}$ from \Cref{cor:worst_case_yz}, and the algorithm just runs the algorithm in \Cref{cor:worst_case_yz}, which implies that
    \[ \Pr_{H \gets \Bias_{q, 1/\lambda^4, \FF_q^s}}\left[\forall x \in \{0, 1\}^{\lambda}, \Pr[(x \| 0^{q-\lambda}, \bfv) \in R_{C_{\lambda}, H}: \bfv \leftarrow \Aa(\ket{z_O}, x)] \geq 1-2^{-\lambda} \right] \geq 1-2^{-\lambda}. \]

We now move to proving Item 2. Suppose for the sake of contradiction that there exists a polynomial $p(\lambda)$, an adversary $\Bb$ which makes $Q(\lambda) = \poly(\lambda)$ oracle queries, and a family of $t(\lambda) = \poly(\lambda)$-bit classical advice $\{z_O\}_{O}$ such that for infinitely many $\lambda$,
    \[ \Pr_{H \gets \Bias_{q, 1/\lambda^4, \FF_q^s}, x \gets \{0, 1\}^{\lambda}}\left[(x \| 0^{q-\lambda}, \bfv) \in R_{C_{\lambda}, H}: \bfv \gets \Bb^{O[H, E_{\lambda}]}(z_O, x)\right] \geq \frac{1}{p(\lambda)}. \]
Our goal will be to arrive at a contradiction by showing that this algorithm $\Bb$ implies a (too good) sampler for $R_{C_\lambda, H}$ that works for infinitely many choices of $\lambda$. Consider the $Q$-query algorithm $\Aa_1^{O[H, E_{\lambda}]}(z_O)$ which samples a random $x \gets \{0, 1\}^{\lambda}$, runs $\Bb^{O[H, E_{\lambda}]}(z_O, x)$ to get $\bfv$, and outputs $(x, \bfv)$. By the definition of $\Bb$, we have
\begin{equation}
    \Pr_{H \gets \Bias_{q, 1/\lambda^4, \FF_q^s}}\left[(x \| 0^{q-\lambda}, \bfv) \in R_{C_{\lambda}, H}: (x, \bfv) \gets \Aa_1^{O[H, E_{\lambda}]}(z_O)\right] \geq \frac{1}{p(\lambda)}\,,
\end{equation}
for infinitely many $\lambda$.  We denote with $\Lambda$ the set of all $\lambda$ for which this bound holds. At a very high level, we will first show that for all but finitely many $\lambda \in \Lambda$, $\Aa_1$'s queries are concentrated on very few points. Once we know that queries to $\Aa_1$ are concentrated on a couple points, replacing the real oracle with an oracle that only contains those few points will give rise to a sampler for many more points than are contained within the oracle itself.  

For each function $H$ and set $S \subseteq \{0, 1\}^{\lambda} \times \Sigma^q$, we define $M_{H, \overline{S}}$ as the query mass that $\Aa_1$ places on points which differ between $O[H, S]$ and $O[H, E]$. Let $\mathsf{SmallSet}_H$ be the event that there exists a list $L_H \subseteq \{0, 1\}^{\lambda} \times \Sigma^q$ such that $|L_H| \leq 2^{\lambda/2}$ and $M_{H, \overline{L_H}} \leq \frac{1}{256p(\lambda)^2 Q(\lambda)}$.

\begin{claim}
    Whenever $\mathsf{SmallSet}_{H}$ does not occur, for all $\ell \leq 2^{\lambda / 2}$, there is an algorithm, $\mathsf{Guesser}$, which makes no queries to an oracle, and outputs a list of $\ell$ distinct points from $R_{C_{\lambda}, H}$ with probability at least $2^{-t} \cdot \left(\frac{1}{256p(\lambda)^2 Q(\lambda)^2}\right)^{\ell}$.
\end{claim}
\begin{proof}
    The algorithm $\mathsf{Guesser}$ is identical to the algorithm in \Cref{fig:guesser}, except starting from $\Aa_1$ instead of a $\QCMA$ verifier for the code intersection subset size problem.  
    
    The proof follows similarly as well.  Let $X_{i}$ be the event that the $i$'th round of $\mathsf{Guesser}(1^{\ell})$ outputs a tuple $(x \| 0^{q-\lambda}, \bfv) \notin \Delta_{i-1}$ such that $(x \| 0^{q-\lambda}, \bfv) \in R_{C_{\lambda}, H}$ and $G$ be the event that the advice is guessed correctly.  By assumption, and because $|\Delta_{i-1}| \leq 2^{\lambda/2}$, we have that $\Aa_1$ places at least $\frac{1}{256p(\lambda)^2 Q(\lambda)}$ query mass on points which differ between $O[H, \Delta_{i-1}]$ and $O[H, E_{\lambda}]$.  Therefore, we have that $\Pr[X_{i} |  X_{1} \land \ldots \land X_{i-1} \land G] \geq \frac{1}{256p(\lambda)^2 Q(\lambda)^2}$.
    
    Applying the chain rule, together with the fact that $\Pr[G] \geq 2^{-t(\lambda)}$, we get that the probability of sampling $\ell$ distinct points from $R_{C_{\lambda}, H}$ is at least
    \begin{equation*}
        \Pr[X_{1} \land \ldots \land X_{\ell} \land G] \geq 2^{-t(\lambda)} \cdot \left(\frac{1}{256p(\lambda)^2 Q(\lambda)^2}\right)^{\ell}\,. \qedhere
    \end{equation*}
\end{proof}
As a corollary, we have that for all but finitely many $\lambda \in \Lambda$, $\mathsf{SmallSet}_{H}$ must occur with high probability.   

\begin{claim}
    There are only finitely many $\lambda \in \Lambda$ such that $\Pr_{H \gets \Bias_{q, 1/\lambda^4, \FF_q^s}}[\mathsf{SmallSet}_{H}] < 1-\frac{1}{4p(\lambda)}$.
\end{claim}
\begin{proof}
Assume for the sake of contradiction that there are infinitely many $\lambda \in \Lambda$ such that $\Pr[\mathrm{SmallSet}_{H}] < 1-\frac{1}{4p(\lambda)}$. Whenever $\mathsf{SmallSet}_{H}$ does not occur, the previous claim gives us a sampler. Thus, when $H$ is sampled from $\Bias_{q, 1/\lambda^4, \FF_q^s}$, the probability of the sampler outputting $\ell$ distinct points from $R_{C_{\lambda}, H}$ is at least
\begin{equation*}
    \Pr_{H \gets \Bias_{q, 1/\lambda^4, \FF_q^s}}\left[\begin{aligned} &\forall i \neq j, (x_i, \bfv_i) \neq (x_j, \bfv_j) \\ &\land (x_i, \bfv_i) \in R_{C_{\lambda}, H} \end{aligned}: \{(x_i, \bfv_i)\}_{i = 1}^{\ell} \gets \mathsf{Guesser}(1^\ell)\right] \geq \frac{1}{4p} \cdot 2^{-t} \cdot \left(\frac{1}{256p(\lambda)^2 Q(\lambda)^2}\right)^{\ell}\,.
\end{equation*}
Taking $\ell = \max\{t, \lambda\} \leq 2^{\lambda/2} \leq 2^{\lambda}$ gives a sampling success probability of 
    \[ \frac{1}{4p(\lambda)} \cdot \left(\frac{1}{512p(\lambda)^2Q(\lambda)^2}\right)^{\ell} \geq \left(\frac{1}{2048p(\lambda)^3Q(\lambda)^2}\right)^{\ell}, \]
which is a contradiction since \Cref{lemma:guesser_hardness} implies this success probability should be at most 
\begin{equation*}
    \left(e^{-\lambda/2}\right)^{\ell} \ll \left(\frac{1}{2048p(\lambda)^3Q(\lambda)^2}\right)^{\ell}\,. \qedhere
\end{equation*}
\end{proof}
Now we know that $\mathsf{SmallSet}_{H}$ occurs with high probability for infinitely many $\lambda \in \Lambda$. For all $H$ for which $\mathsf{SmallSet}_{H}$ occurs, we denote $L^{*}_H$ to refer to any arbitrary set of size at most $2^{\lambda/2}$ such that $M_{H, \overline{L^{*}_H}} \leq \frac{1}{256p(\lambda)^2 Q(\lambda)}$; if $\mathsf{SmallSet}_{H}$ does not occur, then we define $L^{*}_H := \emptyset$. For each function $H$, define the punctured oracle
    \[ O^{*}_H(x, \bfv) = \begin{cases} 1 & \text{if } (x \| 0^{q-\lambda}, \bfv) \in R_{C_{\lambda}, H} \land (x, \bfv) \in L^{*}_H, \\ 0 & \text{otherwise.} \end{cases} \]
Then we have that, conditioned on $\lambda$ being such that $\Pr_{H \gets \Bias_{q, 1/\lambda^4, \FF_q^s}}[\mathsf{SmallSet}_{H}] \geq 1 - \frac{1}{4p(\lambda)}$,
\begin{align*}
    &\Pr_{H \gets \Bias_{q, 1/\lambda^4, \FF_q^s}}\left[(x \| 0^{q-\lambda}, \bfv) \in R_{C_{\lambda}, H}: (x, \bfv) \gets \Aa_1^{O^{*}_H}(z_O) | \mathsf{SmallSet}_{H} \right] \\
    &\hspace{5mm}\geq\Pr_{H \gets \Bias_{q, 1/\lambda^4, \FF_q^s}}\left[(x \| 0^{q-\lambda}, \bfv) \in R_{C_{\lambda}, H}: (x, \bfv) \gets \Aa_1^{O[H, E]}(z_O) | \mathsf{SmallSet}_{H} \right]-\frac{1}{4p(\lambda)} \\
    &\hspace{5mm}\geq\Pr_{H \gets \Bias_{q, 1/\lambda^4, \FF_q^s}}\left[(x \| 0^{q-\lambda}, \bfv) \in R_{C_{\lambda}, H}: (x, \bfv) \gets \Aa_1^{O[H, E]}(z_O) \land \mathsf{SmallSet}_{H} \right]-\frac{1}{4p(\lambda)} \\
    &\hspace{5mm}\geq\frac{1}{p(\lambda)}-\frac{1}{2p(\lambda)} = \frac{1}{2p(\lambda)}\,.
\end{align*}
In the first line, we use the hybrid lemma (\Cref{thm:bbbv}) combined with our bound on the query mass of $\Aa_1$ outside of $L_{H}$. In the second line, we use the definition of conditional probability, together with the fact that $\Pr[\mathsf{SmallSet}_H] \leq 1$. We conclude by using a union bound, together with the fact that $\Pr[\mathsf{SmallSet}_{H}] \geq 1-\frac{1}{4p(\lambda)}$ and the fact that the probability of $\Aa_1$ sampling a point in $R_{C_{\lambda}, H}$ is at least $\frac{1}{p(\lambda)}$ by assumption.

To arrive at a contradiction we construct yet \emph{another} sampler. By the definition of $\Aa_1$, running it produces a uniformly random $x \in \{0, 1\}^{\lambda}$, so there exists an algorithm $\Aa_2$ which, on input $1^{\ell}$, runs $\Aa_1$ $\ell$ times and satisfies
\begin{multline*}
    \Pr_{H \gets \Bias_{q, 1/\lambda^4, \FF_q^s}}\left[\begin{aligned} &\forall i \neq j, (x_i, \bfv_i) \neq (x_j, \bfv_j) \\ &\land (x_i \| 0^{q-\lambda}, \bfv_i) \in R_{C_{\lambda}, H} \end{aligned}: \{(x_i, \bfv_i)\}_{i = 1}^{\ell} \gets \Aa_2^{O^{*}_H}(z_O, 1^\ell)\right] \\ \geq \Pr_{H \gets \Bias_{q, 1/\lambda^4, \FF_q^s}}[\mathsf{SmallSet}_{H}] \cdot \prod_{i = 1}^{\ell} \left(\frac{1}{2p(\lambda)}-\frac{i-1}{2^{\lambda}}\right)\,.
\end{multline*}
Here we applied the definition of conditional probability and used the fact that since the $x_i$'s are uniformly random, the probability that $x_i \in \bigcup_{j=1}^{i-1} \{x_j\}$ is at most $\frac{i - 1}{2^{\lambda}}$.  

As $O^{*}_H$ has at most $2^{\lambda/2}$ nonzero points, we can hardwire $2^{\lambda/2} \cdot (\lambda+\log |\Sigma|^q) \leq 2^{2\lambda/3}$ bits of advice and simulate $O^{*}_H$. By guessing this extra advice along with $z_O$, we get an algorithm $\Aa_3$ such that for $\ell = 2^{3\lambda/4}$,
\begin{align*}
    &\Pr_{H \gets \Bias_{q, 1/\lambda^4, \FF_q^s}}\left[\begin{aligned} &\forall i \neq j, (x_i, \bfv_i) \neq (x_j, \bfv_j) \\ &\land (x_i \| 0^{q-\lambda}, \bfv_i) \in R_{C_{\lambda}, H} \end{aligned}: \{(x_i, \bfv_i)\}_{i = 1}^{\ell} \gets \Aa_3(1^\ell)\right] \\
    &\hspace{5mm}\geq2^{-(t(\lambda)+2^{2\lambda/3})} \cdot \left(1-\frac{1}{4p(\lambda)}\right) \cdot \prod_{i = 1}^{\ell} \left(\frac{1}{2p(\lambda)}-\frac{i-1}{2^{\lambda}}\right) \\
    &\hspace{5mm}\geq 2^{-\ell} \cdot \prod_{i = 1}^{\ell} \left(\frac{1}{2p(\lambda)}-\frac{\ell}{2^{\lambda}}\right) \geq 2^{-\ell} \cdot \left(\frac{1}{3p(\lambda)}\right)^{\ell} = \left(\frac{1}{6p(\lambda)}\right)^{\ell}\,.
\end{align*}
Applying \Cref{lemma:guesser_hardness} for $\ell = 2^{3\lambda/4} \leq 2^{\lambda}$, we have a contradicting upper bound on the success probability of 
\begin{align*}
    \Pr_{H \gets \Bias_{q, 1/\lambda^4, \FF_q^s}}\left[\begin{aligned} &\forall i \neq j, (x_i, \bfv_i) \neq (x_j, \bfv_j) \\ &\land (x_i \| 0^{q-\lambda}, \bfv_i) \in R_{C_{\lambda}, H} \end{aligned}: \{(x_i, \bfv_i)\}_{i = 1}^{\ell} \gets \Aa_3(1^\ell)\right] &\leq \left(1-\frac{1}{\lambda^4}\right)^{\lambda^5 \ell/2} \leq e^{-\lambda\ell/2} \ll \left(\frac{1}{6p(\lambda)}\right)^{\ell}.
\end{align*}
\end{proof}

Having shown that there is a search problem outside of (average-case) $\FBQP/\poly$, we now apply \Cref{lemma:advice_query_measurement}. This was essentially established in \cite{LLPY23}, but we make minor modifications to deal with quantum queries.

\begin{lemma}\label{lemma:bqp_distributional_separation}
    There is a family of distributions $\{\Dd_{\lambda}\}_{\lambda \in \NN}$, where $\Dd_{\lambda}$ is supported on tuples $(G, \Oo')$ of functions $G : \{0, 1\}^\lambda \to \{0, 1\}$ and $\Oo': \{0, 1\}^{\poly(\lambda)} \to \{0, 1\}^{\poly(\lambda)}$, satisfying the following:
    \begin{enumerate}
        \item There is an polynomial-time uniform quantum algorithm $\Aa$ which makes one query to $\Oo'$ such that for all $\Oo'$ there exists a family of $\poly(\lambda)$-qubit quantum advice $\{\ket{z_{\Oo'}}\}_{\Oo'}$ such that
            \[ \Pr_{(G, \Oo') \gets \Dd_{\lambda}}[\forall x \in \{0, 1\}^{\lambda}, \Pr[\Aa^{\Oo'}(\ket{z_{\Oo'}}, x) = G(x)] \geq 1-\negl(\lambda)] \geq 1-\negl(\lambda)\,. \]
        \item For all quantum query algorithms algorithm $\Bb$ that makes $\poly(\lambda)$ queries to $\Oo'$ and receiving a family of $\poly(\lambda)$-bit classical advice $\{z_{\Oo'}\}_{\Oo'}$ depending only on $\Oo'$,
            \[ \Pr_{(G, \Oo') \gets \Dd_{\lambda}, x \gets \{0, 1\}^{\lambda}}[\Bb^{\Oo'}(z_{\Oo'}, x) = G(x)] \leq \frac{3}{5}\,, \]
        for all sufficiently large $\lambda$.
    \end{enumerate}
\end{lemma}
\begin{proof}
Define $\Dd_{\lambda}$ as follows: first, sample a random function $G: \{0, 1\}^{\lambda} \to \{0, 1\}$ and $H \gets \Bias_{q, 1/\lambda^4, \FF_q^s}$ and let 
    \[ \Oo'(x, \bfv) := \begin{cases} G(x) & \text{if } (x \| 0^{q-\lambda}, \bfv) \in R_{C_{\lambda}, H}, \\ \bot & \text{otherwise.} \end{cases} \]

We begin by showing easiness with quantum advice. Let $(\Aa',\{\ket{z'_H}\}_H)$ be the algorithm and advice family from Item 1 of \Cref{lemma:bqp_sampling}. We now construct an algorithm $\Aa$ and family of mixed state advice $\{\rho_{\Oo'}\}_{\Oo'}$.\footnote{This is without loss of generality as a mixed state is a distribution over pure states and so there is always a pure state advice that is at least as good as the mixed state advice.} We describe a randomized procedure to set $\rho_{\Oo'}$ given an oracle $\Oo'$, but in reality, we will set $\rho_{\Oo'}$ to be the mixed state corresponding to the mixture over outputs of this procedure. Sample $(G, H)$ from the distribution of $\Dd_{\lambda}$ conditioned on $\Oo'$; by construction, the joint distribution of $(G, H, \Oo')$ sampled in this procedure is identical to $\Dd_{\lambda}$. We then set our advice to be $\rho_{\Oo'} = \ket{z'_H}$. The algorithm $\Aa$ on input $x$ will run $\bfv \gets \Aa'(\rho_{\Oo'}, x)$, query $(x, \bfv)$ to $\Oo'$, and output whatever $\Oo'$ returns. Item 1 of \Cref{lemma:bqp_sampling} then implies that $\Aa$ is efficient/uniform and that
    \[ \Pr_{(G, \Oo') \gets \Dd_{\lambda}}[\forall x \in \{0, 1\}^{\lambda}, \Pr[\Aa^{\Oo'}(\rho_{\Oo'}, x) = G(x)] \geq 1-\negl(\lambda)] \geq 1-\negl(\lambda). \]

Now suppose for the sake of contradiction that there was some algorithm $\Bb$ which made $Q(\lambda) = \poly(\lambda)$ queries and had $t(\lambda) = \poly(\lambda)$-bit classical advice $\{z_{\Oo'}\}_{\Oo'}$ such that for infinitely many $\lambda$,
    \[ \Pr_{\substack{(G, \Oo') \gets \Dd_{\lambda} \\ x \gets \{0, 1\}^{\lambda}}}[\Bb^{\Oo'}(z_{\Oo'}, x) = G(x)] > \frac{3}{5}. \]
We know that $\Oo'$ returns $G(x)$ only if the query $(x, \bfv) \in R_{C_{\lambda}, H}$. Thus, by a direct reduction to \Cref{lemma:advice_query_measurement}, for a $\frac{1}{4000Q(\lambda)^2}$ fraction of $x \in \{0, 1\}^{\lambda}$, measuring a random query of $\Bb$ to a randomly sampled oracle $\Oo' \gets \Dd_{\lambda}$ will produce $(x, \bfv) \in R_{C_{\lambda}, H}$ with probability at least $\frac{1}{3200Q(\lambda)^2}$.

But now observe that $\Oo'$ can be simulated by querying $G$ and $\Oo := O[H, \{0, 1\}^{\lambda} \times \Sigma^q]$. Thus, for each function $G$, we define the following $Q$-query algorithm $\Bb'[G]$ and classical advice $\{z'_{\Oo}[G]\}_{\Oo}$. First, construct $\Oo'$ from $(G, \Oo)$ (since $\Oo$ uniquely determines $H$) before setting $z'_{\Oo}[G] := z_{\Oo'}$. $\Bb'[G]^{\Oo}(z'_{\Oo}[G], x)$ will run $\Bb^{\Oo'}(z'_{\Oo}[G], x)$ where $\Bb'$ will simulate the oracle $\Oo'$ using its own oracle $\Oo$ and the hardwired oracle $G$ and measure a uniformly chosen query of $\Bb$. As noted earlier, this means that
    \[ \Pr_{G, H, x}[(x, \bfv) \in R_{C_{\lambda}, H}: \bfv \gets \Bb'[G]^{\Oo}(z'_{\Oo}[G], x)] \geq \frac{1}{3200 \cdot  4000 \cdot Q(\lambda)^4} = \frac{1}{\poly(\lambda)}. \]
By taking $G^{*}$ which maximizes the above probability,\footnote{As noted in \cite{LLPY23}, finding $G^{*}$ does not actually require access to the specific $H$ since $\Bb'$ can find $G^{*}$ by itself by using its unbounded computational power to enumerate over all possible $G$ and $\Oo_H$.} $(\Bb'[G^{*}],\{z'_{\Oo}[G^{*}]\}_{\Oo})$ breaks Item 2 of \Cref{lemma:bqp_sampling}.
\end{proof}

A simple diagonalization argument gives us our desired separation (see \Cref{app:diag} for details).
\begin{theorem}\label{thm:advice_separation}
    There is a classical oracle $\Oo$ such that $\BQP^{\Oo}/\qpoly \cap \NP^{\Oo} \cap \coNP^{\Oo} \subsetneq \BQP^{\Oo}/\poly$.\footnote{It is not hard to extend this separation to show that  $\mathsf{YQP}^{\Oo} \cap \NP^{\Oo} \cap \coNP^{\Oo} \subsetneq \BQP^{\Oo}/\poly$, where $\mathsf{YQP}$ is the class of problems that can be decided by a $\BQP$ machine with \emph{untrusted} quantum advice \cite{Aar07, AD14}.}
\end{theorem}

\ifsubmission
\else
\paragraph{Acknowledgments.}
We thank Scott Aaronson, Joe Carolan, Ryan Williams, Rohan Goyal, Venkatesan Guruswami, Mary Wootters and Rachel Zhang for patiently answering our many questions. A special thanks to Anand Natarajan for very helpful discussions in the early stages of this project and Alexandru Gheorghiu and Aparna Gupte for pointing out the efficiency benefits of using biased oracles in the Yamakawa-Zhandry algorithm. JB is supported by Henry Yuen's AFOSR award FA9550-23-1-0363.  VV gratefully acknowledges support from a Simons Investigator Award and a Ford Foundation Chair.
\fi

\bibliographystyle{alpha}
\bibliography{supplemental/refs}

\appendix
\section{Duals of Multiplicity Codes}\label{app:mult}
For a field $\FF_q$, the multiplicity $\mathsf{mult}(f, \alpha)$ of a polynomial $f \in \FF_q[X]$ at a point $\alpha \in \FF_q$ is the largest integer $m$ so that $f^{(i)}(\alpha) = 0$ for any non-negative integer $i < m$. The multiplicity Schwartz-Zippel Lemma from \cite{DKSS13} says that a nonzero degree $k$ univariate polynomial can vanish on at most $k$ points, counting multiplicities.
\begin{lemma}[\cite{DKSS13}]\label{lemma:mult_schwartz_zippel}
    Let $f \in \FF_q[X]$ be a nonzero polynomial of degree at most $k$. Then $\sum_{\alpha \in \FF_q} \mathsf{mult}(f, \alpha) \leq k$.
\end{lemma}

\begin{fact}[Hasse derivatives, see \cite{Bla24}]
    The following properties hold for the Hasse derivative:
    \begin{enumerate}
        \item For any polynomial $f(X) \in \FF_q[X]$, integer $i \geq 0$, and point $\alpha \in \FF_q$, $f^{(i)}(\alpha)$ is the coefficient of $X^i$ in $f(X+\alpha)$.
        \item (Linearity) For any $f, g \in \FF_q[x]$, $\lambda, \mu \in \FF_q$, and $i \geq 0$, $(\lambda \cdot f + \mu \cdot g)^{(i)} = \lambda \cdot f^{(i)} + \mu \cdot g^{(i)}$.
        \item (Product rule) For any $f, g \in \FF_q[x]$ and $i \geq 0$, we have $(f \cdot g)^{(i)} = \sum_{k = 0}^i f^{(k)} \cdot g^{(i-k)}$.
    \end{enumerate}
\end{fact}
We first derive a natural analogue of Lagrange interpolation for the setting of Hasse derivatives.
\begin{lemma}[Hermite interpolation]\label{lemma:hermite_interpolation}
    Let $\FF_q$ be a field, $s \geq 1$ be a positive integer, and $\alpha_1, \ldots, \alpha_n$ be distinct points in $\FF_q$. For $i \in [n]$ and $0 \leq j \leq s-1$, define $\mu_i(X) := \prod_{i' \neq i} (X - \alpha_{i'})^s $ and $\eta_i(X) := (\mu_i(X))^{-1} \bmod{(X-\alpha_i)^s}$.
    
    Then, for all $f(X) \in (\FF_q)_{<sn}[X]$, we can write $f(X) = \sum_{i = 1}^n \sum_{j = 0}^{s-1} A_{i, j}(X) f^{(j)}(\alpha_i)$, where 
    \begin{align*}
        A_{i, j}(X) = \mu_i(X) (X - \alpha_i)^j \sum_{t = 0}^{s-1-j} \eta_i^{(t)}(\alpha_i)(X - \alpha_i)^t.
    \end{align*}
\end{lemma}
\begin{proof}
    We begin by showing that for any $i, i' \in [n]$ and $0 \leq j, j' \leq s-1$, $A_{i, j}^{(j')}(\alpha_{i'}) = 1$ if $(i, j) = (i', j')$ and $0$ otherwise. First, if $i' \neq i$, then $A_{i, j}(X + \alpha_{i'}) = X^s \cdot B_{i, j}(X)$ for some polynomial $B_{i, j}(X)$ so $A_{i, j}^{(j')}(\alpha_{i'}) = 0$ for all $j, j'$. Similarly, since $A_{i, j}(X + \alpha_i) = X^j \cdot C_{i, j}(X)$ for some polynomial $C_{i, j}(X)$, $A_{i, j}^{(j')}(\alpha_i) = 0$ whenever $i = i'$ and $j' < j$. It thus remains to consider $i = i'$ and $j \leq j'$. By the product rule, we know that 
    \begin{align*}
        A_{i, j}^{(j')}(\alpha_i) &= \sum_{k = 0}^{j'} \mu_i^{(j'-k)}(\alpha_i) \cdot \left[\sum_{t = 0}^{s-1-j} \eta_i^{(t)}(\alpha_i)(X - \alpha_i)^{j+t}\right]^{(k)}(\alpha_i) = \sum_{k = j}^{j'} \mu_i^{(j'-k)}(\alpha_i) \cdot \eta_i^{(k-j)}(\alpha_i) \\
        &= \sum_{\ell = 0}^{j'-j}  \mu_i^{(j'-j-\ell)}(\alpha_i) \cdot \eta_i^{(\ell)}(\alpha_i) = (\mu_i \cdot \eta_i)^{(j'-j)}(\alpha_i).
    \end{align*}
    By construction, $(\mu_i \cdot \eta_i)(X) = 1 + h(X) \cdot (X - \alpha_i)^s$ for some polynomial $h(X)$, so $\mu_i(X+\alpha_i) \cdot \eta_i(X+\alpha_i) = 1 + X^s \cdot h(X+\alpha_i)$. As $0 \leq j'-j \leq s-1$, we conclude that $(\mu_i \cdot \eta_i)^{(j'-j)}(\alpha_i)$ equals 1 if $j' = j$ and 0 otherwise.

    Now, let $g(X) = \sum_{i = 1}^n \sum_{j = 0}^{s-1} A_{i, j}(X) f^{(j)}(\alpha_i)$. Note that since $\deg A_{i, j}, \deg f \leq sn-1$, we know that $\deg (g-f) \leq sn-1$. But for any $i' \in [n]$ and $0 \leq j' \leq s-1$, we have that
        \[ (g-f)^{(j')}(\alpha_{i'}) = g^{(j')}(\alpha_{i'})-f^{(j')}(\alpha_{i'}) = \left[\sum_{i = 1}^n \sum_{j = 0}^{s-1} A_{i, j}^{(j')}(\alpha_{i'}) f^{(j)}(\alpha_i)\right] - f^{(j')}(\alpha_{i'}) = f^{(j')}(\alpha_{i'}) - f^{(j')}(\alpha_{i'}) = 0. \]
    Thus, by the multiplicity Schwartz-Zippel Lemma (\Cref{lemma:mult_schwartz_zippel}), $(g-f)(X) = 0$ and so $f(X) = g(X)$.
\end{proof}

\begin{definition}[Generalized multiplicity codes]
    For invertible matrices $U_1, \ldots, U_n \in \FF_q^{s \times s}$, define the generalized multiplicity (GM) code $\mathsf{GM}_{s, \FF_q}(U_1, \ldots, U_n; \alpha_1, \ldots, \alpha_n; k) := \{(U_1 \cdot c_1, \ldots, U_n \cdot c_n): c \in \mathsf{Mult}_{s, \FF_q}(\alpha_1, \ldots, \alpha_n; k)\}$.
\end{definition}
Note that GM codes have distance at least $n-\frac{k-1}{s}$ by \Cref{lemma:mult_schwartz_zippel}.
\begin{theorem}[Duality of GM codes]\label{thm:mult_duality}
    Let $\FF_q$ be a field, and $s \geq 1$ be a positive integer, and $\alpha_1, \ldots, \alpha_n$ be distinct points in $\FF_q$. Then there exist invertible matrices $U_1, \ldots, U_n \in \FF_q^{s \times s}$, so that for any positive integer $k < sn$, 
        \[ \mathsf{Mult}_{s, \FF_q}(\alpha_1, \ldots, \alpha_n; k) = \mathsf{GM}_{s, \FF_q}(U_1, \ldots, U_n; \alpha_1, \ldots, \alpha_n; sn-k)^{\perp}. \]
\end{theorem}
\begin{proof}
    Consider any pair of polynomials $f(X) \in (\FF_q)_{<k}[X]$ and $g(X) \in (\FF_q)_{<sn-k}[X]$. Let $h(X) := f(X) \cdot g(X)$, and note that $h(X)$ has degree at most $sn−2$. By \Cref{lemma:hermite_interpolation}, there exist polynomials $A_{i, j}(X)$ such that
        \[ h(X) = \sum_{i = 1}^n \sum_{j = 0}^{s-1} A_{i, j}(X) h^{(j)}(\alpha_i). \]

    Letting $a_{i, j}$ denote the coefficient of $X^{sn-1}$ in $A_{i, j}(X)$, we see that the coefficient of $X^{sn-1}$ in $h(X)$ is
        \[ \sum_{i = 1}^n \sum_{j = 0}^{s-1} a_{i, j} h^{(j)}(\alpha_i) = \sum_{i = 1}^n \sum_{j = 0}^{s-1} a_{i, j} \sum_{\ell = 0}^j f^{(\ell)}(\alpha_i) g^{(j-\ell)}(\alpha_i) = 0, \]
    by the product rule and the fact that $\deg h \leq sn-2$. Consider the following anti-triangular matrices:
        \[ U_i = \begin{bmatrix} a_{i, 0} & a_{i, 1} & \cdots & a_{i, s-1} \\ a_{i, 1} & \adots & \adots & 0 \\ \vdots & \adots &  0 & \vdots \\ a_{i, s-1} & 0 & \cdots & 0 \end{bmatrix} \in \FF_q^{s \times s}. \]
    If we define $f_i := (f^{(0)}(\alpha_i), \ldots, f^{(s-1)}(\alpha_i)) \in \FF_q^s$ and $g_i := (g^{(0)}(\alpha_i), \ldots, g^{(s-1)}(\alpha_i)) \in \FF_q^s$, we see that
        \[ \sum_{i = 1}^n \langle f_i, U_i \cdot g_i \rangle = \sum_{i = 1}^n f_i^T \cdot U_i \cdot g_i = \sum_{i = 1}^n \sum_{j = 0}^{s-1} a_{i, j} \sum_{\ell = 0}^j f^{(\ell)}(\alpha_i) g^{(j-\ell)}(\alpha_i) = 0. \]
    We claim that $U_i$ are invertible. To see this, note that by \Cref{lemma:hermite_interpolation}, for any $i \in [n]$, 
        \[ A_{i, s-1}(X) = \mu_i(X)(X-\alpha_i)^{s-1} \eta_i(\alpha_i) = \eta_i(\alpha_i)(X-\alpha_i)^{s-1} \prod_{i' \neq i} (X-\alpha_{i'})^s, \]
    so $a_{i, s-1} = \eta_i(\alpha_i)$. As $\mu_i \cdot \eta_i \equiv 1 \bmod{(X-\alpha_i)^s}$, $\mu_i(\alpha_i) \cdot \eta_i(\alpha_i) = 1$ and thus $a_{i, s-1} \neq 0$. Consequently,
        \[ \det(U_i) = a_{i, s-1}^s \cdot \det(J_s) = a_{i, s-1}^s \cdot (-1)^{s(s-1)/2} \neq 0, \]
    where $J_s$ is the $s \times s$ reversal/exchange matrix. Thus, we have shown that $\sum_{i = 1}^n \langle \Enc_C(f)_i, \Enc_{C'}(g)_i \rangle = 0$ for all $f$ and $g$, where $C := \mathsf{Mult}_{s, \FF_q}(\alpha_1, \ldots, \alpha_n; k)$ and $C' := \mathsf{GM}_{s, \FF_q}(U_1, \ldots, U_n; \alpha_1, \ldots, \alpha_n; sn-k)$. We conclude that $C = (C')^{\perp}$ as both $C$ and $(C')^{\perp}$ are vector spaces of dimension $k$.
\end{proof}

\begin{corollary}[\Cref{thm:mult_dual}]
For all parameters $s$, $q$, and $k < sq$, $(\Mult_{s, \FF_q, k})^{\perp}$ has distance at least $\frac{k+1}{s}$.
\end{corollary}
\begin{proof}
    By \Cref{thm:mult_duality}, $\dist((\Mult_{s, \FF_q, k})^{\perp}) = \dist(\mathsf{GM}_{s, \FF_q}(U_1, \ldots, U_q; 1, \ldots, q; sq-k)) \geq q-\frac{sq-k-1}{s} = \frac{k+1}{s}$.
\end{proof}

\section{Diagonalization Arguments}\label{app:diag}
\begin{proof}[Proof of \Cref{thm:proof_separation}]
    The proof is nearly identical to that of \cite{BHNZ25}, but we include it for completeness. Let $\Oo∶ \{0, 1\}^{*} \to \{0, 1\}$ be an oracle and let $\Oo_\lambda$ be the restriction to $\lambda^7$-bit inputs, where the lower threshold is $\lambda_0$ (per \Cref{remark:bin_input}). We define the unary (promise) language $\Ll^{\Oo}$ so that $1^{\lambda} \in \Ll^{\Oo}$ precisely when $\Oo_{\lambda}$ is a YES instance and $1^{\lambda} \notin \Ll^{\Oo}$ precisely when $\Oo_{\lambda}$ is a NO instance.
    
    We consider only oracles $\Oo$ such that each restriction to size $\lambda^7$-bit inputs encodes either a YES or NO instance, and so the containment $\Ll^{\Oo} \in \QMA^{\Oo}$ follows from \Cref{lemma:qma_easiness} (as we can hardcode the values of $\Ll^{\Oo}$ on all inputs of length at most $\lambda_0^7$). Showing that $\Ll^{\Oo} \in \AM^{\Oo}$ for all $\Oo$ is simple: Arthur samples $x \gets \{0, 1\}^{\lambda}$ and Merlin responds with any $\bfv \in \Sigma^q$ such that $(x \| 0^{q-\lambda}, \bfv) \in R_{C_{\lambda}, H}$ (which always exists as $H \in \mathsf{Good}$). Arthur accepts iff $\Oo(x, \bfv) = 1$. Completeness and soundness follow essentially immediately.
    
    Now we prove the lower bound for $\QCMA$ machines. Let $M_1, M_2, \ldots$ be an enumeration of all possible Turing machines. Identify any surjective function $\iota∶ \NN \twoheadrightarrow \NN^2$ and define functions $j, a∶ \NN \to \Nn$ by $(j(\kappa), a(\kappa)) = \iota(\kappa)$. Define $F ∶ \NN \to \NN$ so that $F(a)$ is the minimum value such that for all $\lambda \geq F(a)$, any $Q(\lambda) \leq a\lambda^a$ query algorithm with $t(\lambda) \leq a\lambda^a$-length classical witness must misclassify some $\Oo_{\lambda}$. By \Cref{lemma:qcma_hardness}, for every integer $a$, $F(a)$ is well-defined. We identify integers $n_1, n_2, \ldots$ where the oracles will be defined to be nonzero. Define integers $n_1 := 1+F(a(1)), n_{\kappa} := 1+\max\{ F(a(\kappa)), a(\kappa-1)(n_{\kappa-1})^{a(\kappa-1)}\}$. For any $n \in \NN \setminus \{n_1, n_2, \ldots\}$, let $\Oo$ equal 0 everywhere. For these input lengths, $\Oo_n$ is trivially a NO instance.

    For each $\kappa \in \NN$, run $M_{j(\kappa)}$ on input $1^{n_{\kappa}}$ for $a(\kappa)n_{\kappa}^{a(\kappa)}$ steps and interpret its output as a quantum query circuit $\Aa_{n_{\kappa}}$ which accepts a classical witness. For every query that $\Aa_{n_{\kappa}}$ makes of length $< n_{\kappa}$, use the previously generated definitions of the oracle $\Oo$ to hardcode these answers. For queries $\Aa_{n_{\kappa}}$ makes of length $> n_{\kappa}$, replace the oracle gates with identity circuits. The resulting circuit will be $\Bb_{n_{\kappa}}$, which only makes queries of length $n_{\kappa}$. This new algorithm $\Bb_{n_{\kappa}}$ can be used to derive an oracle $\Oo_{n_{\kappa}}$ by applying \Cref{lemma:qcma_hardness}.

    It remains to prove that no $\QCMA^{\Oo}$ algorithm exists. Assume, for contradiction, that there exists a $\P$-uniform family of oracle circuits $\{\Aa_{\lambda}\}$ that solves the code intersection problem with witnesses of length $t(\lambda) = \poly(\lambda)$ and $Q(\lambda) = \poly(\lambda)$ queries. Then, $\{\Aa_{\lambda}\}$ appears in the Turing machine enumeration as some $M_{j^{*}}$ and there exists some $a^{*}$ such that $t(\lambda), Q(\lambda) \leq a^{*} \lambda^{a^{*}}$. As $\iota$ is a surjection, there exists a $\kappa^{*}$ such that $\iota(\kappa^{*}) = (j^{*}, a^{*})$. Let $\Aa_{n_{\kappa^{*}}}$ be the quantum circuit for inputs of length $n_{\kappa^{*}}$. Since the oracle is defined as being 0 for inputs $\notin \{n_1, n_2, \ldots \}$ and $n_{\kappa^{*}+1} > a^{*} n_{\kappa^{*}}^{a^{*}}$, each query gate for inputs of length $> n_{\kappa^{*}}$ is an identity gate. Thus the circuit $\Bb_{n_{\kappa^{*}}}$ has the exact same output as $\Aa_{n_{\kappa^{*}}}$ on inputs of size $n_{\kappa^{*}}$. However, using \Cref{lemma:qcma_hardness}, we constructed an oracle $\Oo_{n_{\kappa^{*}}}$ that $\Bb_{n_{\kappa^{*}}}$ will misclassify. Therefore, $\Aa_{n_{\kappa^{*}}}$ will answer incorrectly on input $1^{n_{\kappa^{*}}}$, completing the proof.
\end{proof}

\begin{proof}[Proof of \Cref{thm:advice_separation}]
Our proof will closely follow \cite{LLPY23}. Suppose that for each $\lambda$ we generate $(G_{\lambda}, \Oo'_{\lambda}) \gets \Dd_{\lambda}$ and define a language $\Ll^{\Oo'} := \bigsqcup_{\lambda \in \NN} G^{-1}_{\lambda}(1)$ and an oracle $\Oo'$ that returns $\Oo'_{|x|}(x)$ on a query $x \in \{0, 1\}^{*}$. It suffices to show that $\Ll^{\Oo'} \in \BQP^{\Oo'}/\qpoly \cap \NP^{\Oo'} \cap \mathsf{coNP}^{\Oo'}$ and $\Ll^{\Oo'} \notin \BQP^{\Oo'}/\poly$ with probability 1.

To see that $\Ll^{\Oo'} \in \BQP^{\Oo'}/\qpoly$ with probability 1, observe that Item 1 of \Cref{lemma:bqp_distributional_separation} implies that there is a $\BQP$ machine $\Aa^{\Oo'}$ with polynomial-size quantum advice that decides $\Ll^{\Oo'}$ on all $x$ of length $\lambda$ with probability at least $1-\frac{1}{\lambda^2}$ for sufficiently large $\lambda$. As $\sum_{\lambda = 1}^{\infty} \frac{1}{\lambda^2} = \frac{\pi^2}{6} < \infty$, the Borel–Cantelli lemma (\Cref{lemma:borel_cantelli}) implies that $\Aa^{\Oo'}$ decides $\Ll^{\Oo'}$ for all but finitely many $\lambda$ with probability 1. By hard-coding all $x$'s where $\Aa$ and $\Ll$ disagree, $\Aa$ can be modified into a $\BQP/\qpoly$ machine $\Aa'^{\Oo'}$ that decides $\Ll^{\Oo'}$ on all $x \in \{0, 1\}^{*}$ with probability 1.

In addition, for any $x \in \{0, 1\}^{\lambda}$, we can give any $\bfv$ where $(x \| 0^{q-\lambda}, \bfv) \in R_{C_{\lambda}, H}$ to certify $G(x)$. Thus, by a Chernoff/union bound, we know that with probability $1-\negl(\lambda)$ over $\Dd_{\lambda}$, there exists an $\NP/\coNP$ certificate for all $x \in \{0, 1\}^{\lambda}$. We conclude via an identical argument that $\Ll^{\Oo'} \in \NP^{\Oo'} \cap \coNP^{\Oo'}$ with probability 1.

For a $\BQP$ machine $\Bb$ that takes $\poly(\lambda)$-bit classical advice, we define $S_{\Bb}(\lambda)$ to be the event over the choice of $(G, \Oo')$ that there is a $\poly(\lambda)$-bit classical advice family $\{z_{\Oo'}\}_{\Oo'}$ such that
    \[ \Pr[\forall x \in \{0, 1\}^{\lambda}, \Bb^{\Oo'}(z_{\Oo'}, x) = G(x)] \geq \frac{2}{3}. \]
By Item 2 of \Cref{lemma:bqp_distributional_separation}, there exists $\lambda_0 \in \NN$ such that for all $\BQP$ machines $\Bb$, $\Pr_{\Dd_{\lambda}}[S_{\Bb}(\lambda)] \leq \frac{9}{10}$ for all $\lambda \geq \lambda_0$.

We will consider a sequence of input lengths $\lambda_1, \lambda_2, \ldots$ defined by $\lambda_i := T(\lambda_{i-1})+1$, where $T(\lambda)$ is the running time of $\Bb$ on input of length $\lambda$. This means that when $\Bb$'s input length is $\lambda_{i-1}$, it cannot query the oracle on input lengths $\geq \lambda_i$, so it must be the case that
\begin{align*}
    \Pr[S_{\Bb}(\lambda_i) | S_{\Bb}(\lambda_0) \land \ldots \land S_{\Bb}(\lambda_{i-1})] &= \Pr[S_{\Bb}(\lambda_i)] \\
    \implies \Pr[S_{\Bb}(1) \land S_{\Bb}(2) \land \ldots ] \leq \Pr\left[\bigwedge_{i = 0}^{\infty} S_{\Bb}(\lambda_i)\right] &= \prod_{i = 0}^{\infty} \Pr[S_{\Bb}(\lambda_i) | S_{\Bb}(\lambda_0) \land \ldots \land S_{\Bb}(\lambda_{i-1})] \leq \prod_{i = 0}^{\infty} \frac{9}{10} = 0.
\end{align*}
But there are countably many $\BQP$ machines, so $\Pr[\exists \Bb: S_{\Bb}(1) \land S_{\Bb}(2) \land \ldots] = 0$. We conclude that $\Ll^{\Oo'} \notin \BQP^{\Oo'}/\poly$ with probability 1 over the choice of $(G, \Oo')$, as desired.
\end{proof}
\end{document}